\documentclass{jpp}

\usepackage{graphicx}

\usepackage[utf8]{inputenc}
\usepackage[T1]{fontenc}
\usepackage{amsmath,amssymb}
\usepackage{color}
\usepackage{hyperref}
\usepackage{bm}
\usepackage{empheq}
\usepackage{subcaption}

\newtheorem{theorem}{Theorem}[section]
\newtheorem{lemma}[theorem]{Lemma}
\newtheorem{corollary}[theorem]{Corollary}
\newtheorem{proposition}[theorem]{Proposition}

\newtheorem{definition}[theorem]{Definition}
\newtheorem{hypotheses}[theorem]{Hypotheses}

\newcommand{\abs}[1]{\left| #1 \right|}
\newcommand{\norm}[1]{\left\lVert #1 \right\rVert}
\newcommand{\ip}[2]{\left\langle #1 , \ #2 \right\rangle }

\newcommand{\dd}[2]{\frac{\mathrm{d} #1}{\mathrm{d} #2}}
\newcommand{\pd}[2]{\frac{\partial #1}{\partial #2}}

\newcommand{\dif}{\hspace{1pt} \mathrm{d}}
\newcommand{\df}{\mathrm{d}}

\DeclareMathOperator{\Diag}{Diag}

\DeclareMathOperator{\sech}{sech}

\let\Im\relax
\DeclareMathOperator{\Im}{Im}

\newcommand{\Rbb}{\mathbb{R}}
\newcommand{\Tbb}{\mathbb{T}}
\newcommand{\Zbb}{\mathbb{Z}}

\newcommand{\Bbf}{\bm{B}}
\newcommand{\Fbf}{\bm{F}}
\newcommand{\Jbf}{\bm{J}}

\newcommand{\bbf}{\bm{b}}
\newcommand{\ebf}{\bm{e}}
\newcommand{\hbf}{\bm{h}}
\newcommand{\nbf}{\bm{n}}
\renewcommand{\qbf}{\bm{q}}
\newcommand{\rbf}{\bm{r}}
\newcommand{\tbf}{\bm{t}}

\newcommand{\vbf}{\bm{v}}
\newcommand{\xbf}{\bm{x}}

\newcommand{\Asf}{\mathsf{A}}
\newcommand{\Bsf}{\mathsf{B}}
\newcommand{\Csf}{\mathsf{C}}
\newcommand{\Dsf}{\mathsf{D}}

\newcommand{\Ccal}{\mathcal{C}}
\newcommand{\Fcal}{\mathcal{F}}
\newcommand{\Ocal}{\mathcal{O}}

\newcommand{\xibf}{\bm{\xi}}

\newcommand{\I}{\mathrm{i}}

\newcommand{\phiic}{\phi^{(\text{IC})}}

\newcommand{\ac}[1]{a^{(#1)}}

\newcommand{\Nrho}{N_\rho}
\newcommand{\Ns}{N_s}

\usepackage[normalem]{ulem}
\newcommand{\edit}[2]{#2}

\shorttitle{Magnetic NAE: Ill-Posedness and Regularization}
\shortauthor{M. Ruth, R. Jorge and D. Bindel}

\title{The High-Order Magnetic Near-Axis Expansion: Ill-Posedness and Regularization}

\author{Maximilian Ruth\aff{1}\corresp{\email{maximilian.ruth@austin.utexas.edu}}, https://orcid.org/0000-0003-1179-2472, Rogerio Jorge\aff{2}, https://orcid.org/0000-0003-2941-6571 \and David Bindel\aff{3}, https://orcid.org/0000-0002-8733-5799}

\affiliation{\aff{1}Center for Applied Mathematics, Cornell University,
Ithaca, NY 14853, USA
\aff{2}Department of Physics, University of Wisconsin-Madison, WI 53706, USA
\aff{3}Department of Computer Science, Cornell University,
Ithaca, NY 14853, USA}

\begin{document}

\maketitle

\begin{abstract}
When analyzing stellarator configurations, it is common to perform an asymptotic expansion about the magnetic axis. 
This so-called near-axis expansion is convenient for the same reason asymptotic expansions often are, namely, it reduces the dimension of the problem. 
This leads to convenient and quickly computed expressions of physical quantities, such as quasisymmetry and stability criteria, which can be used to gain further insight. 
However, it has been repeatedly found that the expansion diverges at high orders \edit{}{in the distance from axis}, limiting the physics the expansion can describe. 
In this paper, we show that the near-axis expansion diverges in vacuum due to ill-posedness and that it can be regularized to improve its convergence. 
Then, using realistic stellarator coil sets, we \edit{show that the near-axis expansion can converge to ninth order in the magnetic field, giving accurate high-order corrections to the computation of flux surfaces}{demonstrate numerical convergence of the vacuum magnetic field and flux surfaces to the true values as the order increases}.
We numerically find that the regularization improves the solutions of the near-axis expansion under perturbation, and we demonstrate that the radius of convergence of the vacuum near-axis expansion is correlated with the distance from the axis to the coils.
\end{abstract}


\section{Introduction}
The design of stellarators is a computationally intensive task.
The most basic problem in stellarator design -- that of computing the magnetic field -- requires solving the steady-state magnetohydrostatics equations (MHS).
These equations are difficult to solve for reasons familiar to many problems in physics: they are nonlinear and three-dimensional.
Popular MHS equilibrium solvers include VMEC \citep{Hirshman1983}, DESC \citep{dudt2020}, and SPEC \citep{Hudson2012}, all of which take on the order of seconds to minutes to compute a single equilibrium.
Beyond equilibrium solving, there are potentially many other stellarator objectives that are expensive to compute, with plasma stability metrics being a major example.
When optimizing for stellarators, the costs of equilibrium and objective solving can limit the speed of the overall design process.
This, in combination with the high dimensionality of specifying 3D fields, motivates a need for simpler alternatives.

Recently, the \textit{near-axis expansion} \citep{mercier_equilibrium_1964,solovev_plasma_1970} has gained traction as an alternative to full 3D MHS solvers.
The near-axis expansion works by asymptotically expanding all of the relevant plasma variables (such as magnetic coordinates, pressure, rotational transform, and plasma current) in the distance from the magnetic axis, which is assumed to be small relative to a characteristic magnetic scale length.
The resulting equations are a hierarchy of one-dimensional ODEs, which can be solved orders of magnitude faster than 3D equilibria.
This allows for one to quickly find large numbers of optimized stellarators \citep{landreman-mapping-2022,giuliani_direct_2024}, something that was previously unavailable to the stellarator community.

In addition to the speed, the near-axis expansion has other benefits.
For instance, in \citet{garren1991} it was shown that quasisymmetry imposes more constraints than free parameters in the expansion, leading to the conjecture that non-axisymmetric but perfectly quasisymmetric stellarators cannot exist. 
Many objectives have been defined and computed for the near-axis expansion, including quasisymmetry \citep{Landreman2019}, quasi-isodynamicity \citep{mata_direct_2022}, isoprominence \citep{burby2023}, and Mercier and magnetic-well conditions for stability \citep{landreman_well_2020,kim2021}.
There is evidence that other higher-order effects such as ballooning and linear gyrokinetic stability could be investigated as well \citep{jorge-turbulence-2020}.
The near-axis expansion has also been combined with a type of quadratic flux minimizing surfaces and coil optimization to create free-field optimized QA equilibria \citep{giuliani_direct_2024}.
In sum, the connection between easily expressed objectives, a relatively low-dimensional equilibrium description, and fast computation has led to the increased use \edit{axis expansion}{of the near-axis expansion}.

However, the near-axis expansion is not without drawbacks.
The primary drawback is fundamental: the expansion has limited accuracy far from the axis.
For instance, in the ``far-axis'' regime, there can be large errors in the magnetic shear and magnetic surfaces can self-intersect \citep{landreman-mapping-2022}.
The paper by \citet{jorge-turbulence-2020} also indicates that higher-order terms may be needed for stability; such as magnetic curvature terms.
Unfortunately, attempts to use higher-order terms have resulted in divergent asymptotic series, limiting the accuracy to small plasma volumes.
Most series go to first, second, or sometimes third order in the distance from the axis in the relevant quantities, with any more terms typically reducing accuracy rather than improving it.
Therefore, if we want to include more physics objectives over larger volumes in the near-axis expansion, we must overcome the issue of series divergence.

Unfortunately, the issue of divergence is confounded by many of the assumptions that can be incorporated into the near-axis expansion. 
The most extreme case is that of QA stellarators, where it has been shown that the system of equations for QA is overdetermined beyond third order in the expansion. 
Obviously, unless one relaxes the problem \citep[e.g.~via anisotropic pressure;][]{rodriguez2021}, one cannot generally ask for a convergent QA near-axis expansion in such a circumstance.
In the simpler case of non-quasisymmetric stellarators with smooth pressure gradients and nested surfaces, it is still unknown whether there are non-axisymmetric solutions to MHS \citep{grad1967,constantin2021}. 
Recent work has found that perturbing for small force \citep{constantin_quasisymmetric_2021} or non-flat metrics \citep{cardona_asymmetry_2024} allow for integrable solutions, but currently, there is no guarantee of solutions of MHS, let alone convergent asymptotic expansions.

So, to begin the task of building convergent numerical methods for the near-axis expansion, we focus on a problem we know is solvable: Laplace's equation for a vacuum magnetic potentials following \citet{jorge-arbitrary-order-2020}.
This can be solved in direct (Mercier) coordinates \citep{mercier_equilibrium_1964} with no assumption of nested surfaces.
Additionally, because solutions of Laplace's equation are real analytic, there exist near-axis expansions of the equation that converge within a neighborhood of the axis.
Despite these guarantees, even the near-axis expansion of Laplace's equation diverges.

In this paper, we show that the vacuum near-axis expansion diverges for a reason: \edit{namely that }{}Laplace's equation as a near-axis expansion is ill-posed (\S\ref{sec:ill-posedness}, following background in \S\ref{sec:background}).
To address this issue, we introduce a small regularization term to Laplace's equation and expand to find a regularized near-axis expansion.
We do this by including a viscosity term to Laplace's equation that damps the highly oscillatory unstable modes responsible for the ill-posedness.
By appropriately bounding the input of the near-axis expansion in a Sobolev norm, we prove that this term results in a uniformly converging near-axis expansion within a neighborhood of the axis.

Following the theory, we describe a pseudo-spectral method for finding solutions to the near-axis expansion to arbitrary order in \S\ref{sec:numerical-method}.
In \S\ref{sec:examples}, we use the numerical method to show two examples of high-order near-axis expansions: the rotating ellipse and Landreman-Paul \citep{landreman-paul}.
We find that the near-axis expansion magnetic field, rotational transform, and magnetic surfaces can converge accurately near the axis for unperturbed initial data.
The region of convergence is observed to be dictated by the distance from the magnetic axis to the coils.
Then, by perturbing the on-axis inputs, we show that the regularized expansion obeys Laplace's more accurately farther from the axis.
Finally, we conclude in \S\ref{sec:conclusion}. 

\section{Background}
\label{sec:background}
In this section, we introduce the near-axis expansion for the vacuum field equilibrium problem. 
This presentation follows closely with \citet{jorge-arbitrary-order-2020}. 
We begin with a discussion of the geometry of the near-axis problem (\S\ref{subsec:NA-Geometry}), introduce the near-axis expansion (\S\ref{subsec:nae}), define the magnetic field problem (\S\ref{subsec:Vacuum-Fields}), and finally discuss finding straight field-line magnetic coordinates (\S\ref{subsec:Magnetic-Coordinates}).
For a more full discussion of the near-axis expansion to all orders, including with pressure gradients, we recommend the papers by \citet{jorge-arbitrary-order-2020} and \citet{sengupta_stellarator_2024}.
For ease of reference, we have summarized the equations in the background in Box \eqref{eq:nae-box} for the magnetic field and Box \eqref{eq:fieldline-box} for magnetic coordinates.

\subsection{Near-Axis Geometry}
\label{subsec:NA-Geometry}
We define a magnetic axis as a $C^\infty$ closed curve $\rbf_0 : \Tbb \to \Rbb^3$ with $\rbf_0' \neq 0$ and a nonzero tangent magnetic field (see \S\ref{subsec:Vacuum-Fields} for details about the field). 
We define a near-axis domain about the axis with radius $\sigma$ as
\begin{equation*}
    \Omega_\sigma = \left\{\rbf \in \Rbb^3 \mid  \forall s \in \Tbb, \ \norm{\rbf - \rbf_0(s)} < \sigma \right\}.
\end{equation*}
We note that the assumption that the axis is infinitely differentiable is necessary for the near-axis expansion to formally go to arbitrary order, and we will eventually reduce the required regularity for the inputs of the regularized expansion, summarized in Box \eqref{eq:reg-box}.

\begin{figure}
    \centering
    \includegraphics[width=0.6\linewidth]{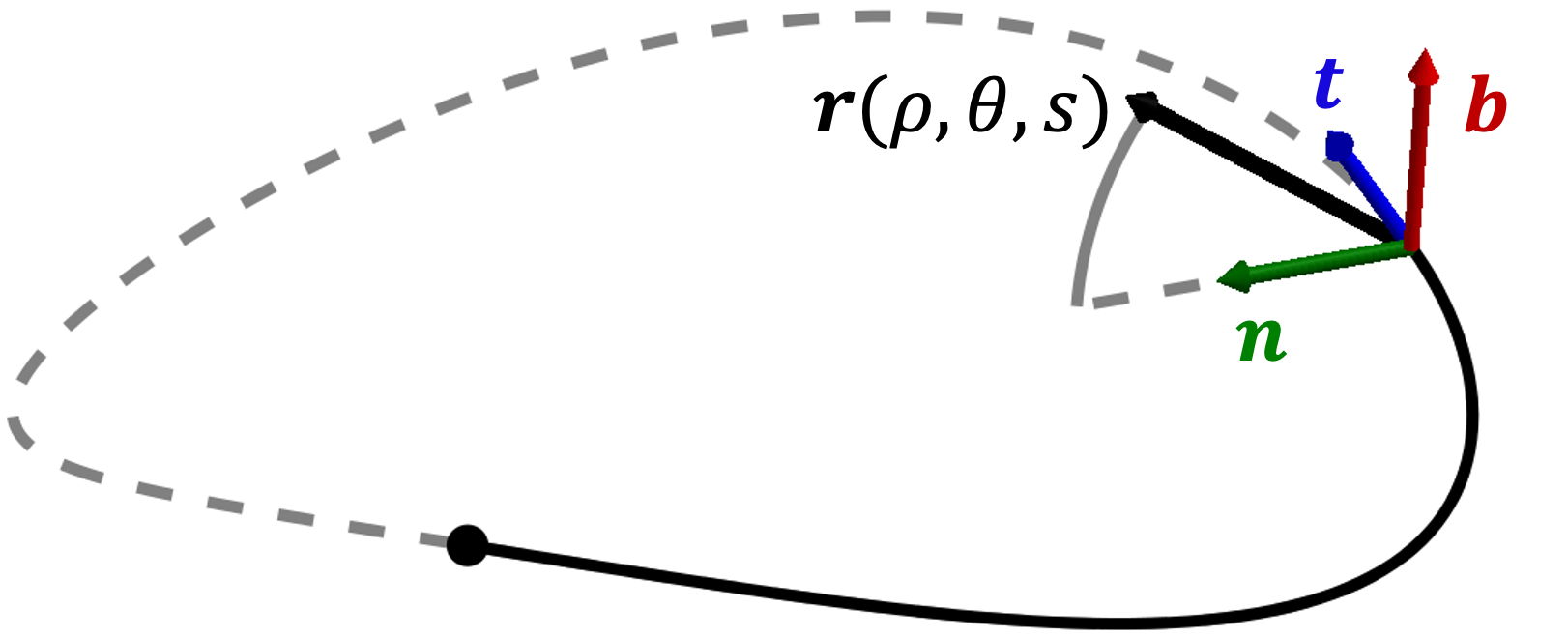}
    \caption{Schematic of the direct near-axis Frenet-Serret coordinate frame.}
    \label{fig:frenet-serret}
\end{figure}
We define a direct coordinate system $\rbf : \Omega_\sigma^0 \to \Omega_\sigma$ where $(\rho,\theta,s) \in \Omega_\sigma^0 = [0,\sigma) \times \Tbb^2$ is the solid torus as (see Fig.~\ref{fig:frenet-serret})
\begin{equation}
\label{eq:general-coord}
    \rbf(\rho, \theta, s) = \rbf_0(s) + Q(s) \begin{pmatrix}
        x \\ y \\ 0
    \end{pmatrix}, \qquad 
        (x, \ y) = (\rho \cos \theta, \  \rho \sin \theta ),
\end{equation}
where $Q$ is an orthonormal basis for the local coordinates at the axis with the tangent vector in the third column, i.e.~$\tbf = \rbf_0' / \ell' = Q\ebf_3$ where we assume $\ell' = \norm{\bm r_0'} > 0$. 
Both the $\xbf = (x,y,s)$ and the $\qbf = (\rho,
\theta,s)$ coordinate frames are useful, as $\xbf$ is non-singular with respect to the coordinate transformation while $\qbf$ diagonalizes the near-axis expansion operator.
To perform calculus in the $\qbf$ basis, we require the induced metric from $\Rbb^3$. 
For this, we first compute the coordinate derivative
\begin{equation}
\label{eq:general-F}
    F = \dd{\rbf}{\qbf} = Q \begin{pmatrix}
        \cos \theta & - \rho \sin \theta & 0\\ 
        \sin \theta & \rho \cos \theta & 0\\
        0 & 0 & \ell' 
    \end{pmatrix} + Q K \xbf \ebf_3^T,
\end{equation}
where $K$ is an antisymmetric matrix determining the derivative of $Q$ in the near-axis basis $Q'(s) = Q K$.
Using this matrix, the metric is computed as 
\begin{equation*}
    g = F^T F.
\end{equation*}

For the numerical examples in this paper, we specifically consider the Frenet-Serret frame, meaning the local basis $Q$ and its derivative are defined by
\begin{equation*}
    Q = \begin{pmatrix}
        | & | & | \\
        \nbf & \bbf & \tbf \\
        | & | & |
    \end{pmatrix}, \qquad K = \ell' \begin{pmatrix}
        0 & -\tau & \kappa \\ 
        \tau & 0 & 0 \\
        -\kappa & 0 & 0
    \end{pmatrix},
\end{equation*}
where $\kappa = \norm{\tbf'}/\ell'$ is the axis curvature, $\nbf = \norm{\tbf'}/(\kappa \ell')$ is the normal vector, $\bbf = \tbf \times \nbf$ is the binormal vector, and $\tau = -\norm{\bbf'}/\ell'$ is the axis torsion.
Alternative forms of the curvature and torsion are
\begin{equation*}
    \kappa = \frac{\norm{\rbf_0' \times \rbf_0''}}{(\ell')^3}, \qquad \tau = \frac{(\rbf_0' \times \rbf_0'') \cdot \rbf_0'''}{\norm{\rbf_0' \times \rbf_0''}^2}.
\end{equation*}
For the Frenet-Serret coordinate system to be well-defined and non-singular on $\Omega_\sigma^0$, we require $\sigma^{-1} > \kappa > 0$.
In particular, no straight segments are allowed in the Frenet-Serret frame, disallowing quasi-isodynamic (QI) stellarators \citep{mata_direct_2022}. 
An alternative choice for axis coordinates that allow for straight segments is Bishop's coordinates \citep{bishop1975,duignan_normal_2021}.

Replacing the Frenet-Serret basis into \eqref{eq:general-F}, we obtain
\begin{equation*}
    F = Q \begin{pmatrix}
        \cos \theta & - \rho \sin \theta & -\ell' \tau \rho \sin \theta \\ 
        \sin \theta & \rho \cos \theta  & \ell' \tau \rho \cos \theta \\
        0 & 0 & h_s
    \end{pmatrix},\qquad 
    g = \begin{pmatrix}
        1 & 0 & 0 \\
        0 & \rho^2 & \ell' \tau \rho^2 \\
        0 & \ell' \tau \rho^2 & (\ell')^2\tau^2 \rho^2 + h_s^2
    \end{pmatrix},
\end{equation*}
where 
\begin{equation*}
    h_s(\rho,\theta,s) = \ell' (1 - \kappa \rho \cos \theta).
\end{equation*}
The local volume ratio is given by
\begin{equation*}
    \sqrt g = \det F = \rho h_s,
\end{equation*}
and the inverse metric is 
\begin{equation*}
    g^{-1} = \begin{pmatrix}
        1 & 0 & 0 \\
        0 & (\ell')^2 \tau^2 h_s^{-2} + \rho^{-2} & -\ell' \tau h_s^{-2} \\
        0 & -\ell' \tau h_s^{-2} & h_s^{-2}
    \end{pmatrix}.
\end{equation*}

To find metrics associated with the more general coordinate system \eqref{eq:general-coord}, we consider transformations of the form $(\rho,\theta,s) \mapsto (\rho,\omega(\theta,s),s)$ where $\omega(\theta,s) =\theta + \lambda(s)$. 
This transformation rotates the orthonormal frame, changing the metric to
\begin{equation*}
        g = \begin{pmatrix}
        1 & 0 & 0 \\
        0 & \rho^2 & \ell' T \rho^2 \\
        0 & \ell' T \rho^2 & (\ell')^2T^2 \rho^2 + h_s^2(\rho,\omega-\lambda,s),
    \end{pmatrix}, \qquad T = \tau - \frac{\lambda'}{\ell'}.
\end{equation*}
In this way, the general set of near-axis frames can be represented by a simple replacement of $\tau$ with $T$.
A special case of this transformation is when $T=0$, yielding \citep{mercier_equilibrium_1964}
\begin{equation*}
    \lambda = \int_{0}^s \tau(s) \ell'(s) \dif s.
\end{equation*}
In this coordinate frame, the metric becomes diagonal: $g = \Diag(1, \rho^2, h_s^2(\rho,\omega-\lambda,s))$. 
The fact that $g$ is diagonalized is convenient for theoretical manipulations, but variables expressed in terms of $\omega$ are multivalued for curves with non-integer total torsion.
This unfortunate consequence is important for numerical methods, as Fourier series can not be used in $s$, and additional consistency requirements are necessary. 
This is part of the motivation for using the Frenet-Serret frame for the numerical examples herein.

\subsection{The Near-Axis Expansion}
\label{subsec:nae}
Now, we consider expansions of functions $A \in C^\infty(\Omega_\sigma^0)$ about the magnetic axis.
We formally expand in small distances from the axis $\rho \ll \min \kappa^{-1}$ as
\begin{gather*}
    A(\rho, \theta, s) = \sum_{m=0}^{\infty}A_m(\theta,s) \rho^m, \qquad A_m(\theta, s) = \sum_{n=0}^m A_{mn}(s) e^{(2n-m)\I\theta}, \\
    A_{mn}(s) = \sum_{k\in\Zbb} A_{mnk}e^{\I k s}.
\end{gather*}
This expansion is not guaranteed to converge anywhere for $A\in C^\infty$, but it is asymptotic to $A$ near the axis, i.e.
\begin{equation*}
    \abs{A - A_{<m}} = \Ocal(\rho^{m}),
\end{equation*}
where we define $A_{<m}$ as the partial sum
\begin{equation*}
    A_{<m} = \sum_{n=0}^{m-1} A_n \rho^n.
\end{equation*}
In defining magnetic coordinates, we also find it convenient to expand $A$ in $\bm x$ as
\begin{equation*}
    A(x,y,s) = \sum_{\mu=0}^{\infty} \sum_{\nu = 0}^\mu A_{\mu \nu}(s) x^{\mu-\nu} y^\nu,
\end{equation*}
where we use Latin indices for the $\qbf$ frame and Greek indices for the $\xbf$ frame. 
If we require $A$ to be real-analytic on $\Omega_\sigma^0$, there additionally exists a $\sigma' < \sigma$ so that the the asymptotic series converges uniformly on $\Omega_{\sigma'}^0$ (this does not necessarily extend to all of $\Omega_\sigma^0$). 

Throughout this paper, we attempt to minimize the number of complicated summation formulas resulting from the near-axis expansion (NAE). 
For instance, if we have two functions $A, B\in C^{\infty}$ and we want the $m$th component of the series of $C=AB$, we will write $C_{m} = (AB)_{m}$, rather than
\begin{equation*}
    C_{m} = \sum_{\ell = 0}^m A_{m-\ell} B_\ell.
\end{equation*}
As expressions become increasingly complicated, this notation provides a concise description of the mathematics involved.
In Appendix \ref{app:operations}, we define a number of relevant operations on series that the interested reader can use to expand the expressions within this paper.
In \S\ref{sec:numerical-method}, we discuss how this is similarly convenient for the purposes of programming NAE operations.
Rather than implementing residuals via complicated summation formulas, the operations in Appendix \ref{app:operations} are called, allowing for a simple framework for developing new code.

An important exception to the general rule of condensing notation is in defining any linear operators that must be inverted through the course of an asymptotic expansion. 
Detailed understanding of such operators are necessary for both numerical implementation and analysis on the series.

\subsection{Vacuum Fields}
\label{subsec:Vacuum-Fields}
In steady state, the vacuum magnetic field $\Bbf\in C^\infty(\Omega_\sigma^0)$ satisfies
\begin{equation*}
    \nabla \cdot \Bbf = 0, \qquad \Jbf = \nabla \times \Bbf = 0,
\end{equation*}
where $\Jbf$ is the plasma current density.
The fundamental near-axis assumption is that $\Bbf$ evaluated on the axis is nonzero and parallel to the magnetic axis, i.e.~for some $B_0 \in C^{\infty}(\Tbb)$
\begin{equation*}
    \Bbf(0,\theta,s) = B_0(s) \tbf(s).
\end{equation*}
Off the axis, we express the magnetic field on $\Omega_\sigma^0$ as
\begin{equation}
\label{eq:magnetic-scalar-potential}
    \Bbf = \nabla \phi + \tilde{\Bbf}, \qquad \tilde{\Bbf} = \nabla \left[ \int_{0}^s B_0(s') \ell'(s') \dif s'\right] = B_0(s)\ell'(s) \nabla s,
\end{equation}
where $\phi\in C^{\infty}(\Omega_\sigma^0)$ is the magnetic potential satisfying $\left.\nabla \phi\right|_{\rho=0} = 0$.
Then, taking the divergence of \eqref{eq:magnetic-scalar-potential}, we find the magnetic scalar potential satisfies Poisson's equation
\begin{equation}
\label{eq:poisson}
    \Delta \phi = - \nabla \cdot \tilde{\Bbf}.
\end{equation}

By construction, the field $\Bbf$ in \eqref{eq:magnetic-scalar-potential} is locally the gradient of some function, so it is curl-free. 
This means the contributions from $\tilde{\Bbf}$ in \eqref{eq:magnetic-scalar-potential} can locally be absorbed to recover Laplace's equation for the potential.
However, because $\Omega_\sigma$ is not simply connected and $\bm B \cdot \bm t$ is single-valued on axis, the closed-loop axis integral $\oint \Bbf \cdot \df \bm \ell$ demonstrates that it is not possible for $\Bbf$ to be globally the gradient of a single-valued function $\phi$. 
In contrast, the Poisson's equation formulation contains only single-valued functions, which is convenient both numerically and analytically.

In coordinates, we can write the magnetic field in \eqref{eq:magnetic-scalar-potential} as
\begin{equation*}
    B^i = g^{ij} \pd{\phi}{q^j} + \tilde{B}^i, \qquad \tilde{B}^i = B_0 \ell' g^{ij}\pd{s}{q^j},
\end{equation*}
where we assume summation over repeated indices and $g_{ij}$ and $g^{ij}$ are the components of the metric and inverse metric respectively. 
Then, Poisson's equation in coordinates becomes
\begin{equation}
\label{eq:laplace-coordinate}
    \sqrt{g}^{-1}\pd{}{q^i}\left( \sqrt{g} g^{ij} \pd{\phi}{q^j} \right) = - \sqrt{g}^{-1} \pd{}{q^i}\left(\sqrt{g} \tilde{B}^i \right).
\end{equation}
Multiplying this by $h_s$ and expanding this in the $\bm q$ coordinate system, we have
\begin{multline}
\label{eq:laplace-unexpanded}
    \frac{1}{\rho^2}\left[ \rho \pd{}{\rho}\left(\Asf \rho \pd{\phi}{\rho} \right)
    + \pd{}{\theta}\left(\Bsf \pd{\phi}{\theta}\right)\right] =  \\ - \pd{}{\theta}\left( \Csf \left( \pd{\phi}{s} + B_0 \ell'\right) \right) - \pd{}{s}\left( \Csf \pd{\phi}{\theta} \right) - \pd{}{s} \left(\Dsf \left(\pd{\phi}{s} + B_0 \ell' \right)  \right),
\end{multline}
where 
\begin{align*}
    \Asf &= h_s, & \Bsf &= h_s + \rho^2 h_s^{-1} (\ell')^2 \tau^2 , & \Csf &= -h_s^{-1} \ell' \tau, & \Dsf &= h_s^{-1}.
\end{align*}

From here, we can substitute the asymptotic expansions of $\phi$ and the coefficients $\Asf$, $\Bsf$, $\Csf$, and $\Dsf$ into \eqref{eq:laplace-unexpanded}.
At each order in $\rho$, Poisson's equation becomes
\begin{align}
\label{eq:laplace-near-axis}
    \ell' \left(m^2 + \pd{^2}{\theta^2}\right) \phi_{m} =& - (\nabla \cdot \tilde{B})_{m-2} -(\Delta \phi_{<m})_{m-2}\\
\nonumber
    =&- m(m-1) \Asf_1 \phi_{m-1} - \left[\pd{}{\theta}\left(\Bsf \pd{\phi_{<m}}{\theta} \right)\right]_m \\
\nonumber
    &- \left[ \pd{}{\theta}\left( \Csf \left(\pd{\phi_{<m}}{s} + B_0 \ell' \right) \right) + \pd{}{s}\left( \Csf \pd{\phi_{<m}}{\theta} \right)\right]_{m-2} \\
\nonumber
    &- \left[\pd{}{s} \left(\Dsf \left( \pd{\phi_{<m}}{s} + B_0 \ell'\right)  \right)\right]_{m-2},
\end{align}
where $\Asf_1 = (h_s)_1 = -\ell' \kappa \cos \theta$. The right-hand side of \eqref{eq:laplace-near-axis} does not depend on orders of $\phi$ higher than $\phi_{m-1}$, so inverting the left-hand-side operator gives an iteration for obtaining $\phi_m$ at each order.

However, the operator $\ell' (m^2 + \pd{^2}{\theta^2}) $ is singular, so we must confirm there are no secular terms.
Specifically, we always have an unknown homogeneous solution at $\Ocal(m)$ of the form $\phi_{m0} e^{-i m \theta} + \phi_{mm} e^{i m \theta}$. 
For this, we use Fredholm's alternative, which states that the \eqref{eq:laplace-near-axis} is solvable if the right-hand side is orthogonal to the null space of the adjoint of the operator on the left-hand side. 
Because the operator is self-adjoint, the right-hand side must be orthogonal to $e^{\pm i m \theta}$, i.e.
\begin{multline*}
    \left\langle m(m-1) \Asf_1 \phi_{m-1} + \left[\pd{}{\theta}\left(\Bsf_{>0} * \pd{\phi_{<m}}{\theta} \right)\right]_m \right. + \left[ \pd{}{\theta}\left( \Csf \left( \pd{\phi_{<m}}{s} + B_0 \ell'\right) \right)\right]_{m-2}\\
      + \left.\left[\pd{}{s}\left( \Csf \pd{\phi}{\theta} \right) + \pd{}{s} \left(\Dsf \left( \pd{\phi_{<m}}{s} + B_0 \ell'\right)  \right)\right]_{m-2}, \ e^{\pm i m \theta}\right\rangle=0,
\end{multline*}
where the inner product is defined by
\begin{equation*}
    \ip{f}{g} = \int_{0}^{2\pi} f(\theta)\overline{g}(\theta) \dif \theta.
\end{equation*}
The $m-2$ coefficient of any analytic function is orthogonal to $e^{i m \theta}$, so we can remove the $\Csf$ and $\Dsf$ terms. 
The same argument allows us to remove the torsion terms in $\Bsf$. 
This mean only contributions from $\Asf_1 = \Bsf_1 = - \ell' \kappa \cos \theta$ and $\phi_{m-1}$  survive, so we only need to verify
\begin{align*}
    \ip{m(m-1) \phi_{m-1} \cos \theta + \left[\pd{}{\theta}\left(\pd{\phi_{m-1}}{\theta} \cos \theta \right)\right]_m }{e^{\pm i m \theta}} = 0.
\end{align*}
Using the identity $\cos \theta = (e^{i \theta} + e^{-i \theta})/2$, a quick calculation confirms the above identity holds.

Now, we consider the problem where $\phi$ is unknown.
The fact that the Fredholm condition is automatically satisfied at each order implies that $\phi_{m0}$ and $\phi_{mm}$ are free parameters at each order in the near-axis expansion. 
So, these coefficients are an infinite-dimensional set of initial conditions for the near-axis expansion of Poisson's equation.
Intuitively, one can think of the coefficients $\phi_{m0}$ and $\phi_{mm}$ as specifying the Fourier coefficients of an infinitely thin tube about the magnetic axis.
In this way, the imposition of conditions at each order compensates for the fact that PDEs typically satisfy conditions on co-dimension 1 surfaces, whereas the near-axis expansion is specified on a co-dimension 2 curve. 

In addition to $\phi_{m0}$ and $\phi_{mm}$ as free parameters, we also treat $\rbf_0$ and $B_0$ as inputs to the near-axis problem. 
The requirement that the magnetic field be tangent to the axis with magnitude $B_0$ results in a constraint that $\phi_{00} = \phi_{10} = \phi_{11} = 0$.
In total, the direct vacuum near-axis problem can be written as

\begin{empheq}[box={\fboxsep=6pt\fbox}]{equation}
\label{eq:nae-box}
\begin{aligned}
     &\text{input: }&& \text{axis } \rbf_0\in C^{\infty} \text{, on-axis field } B_0\in C^\infty, \\
    &&&                \text{higher moments } \phi_{m0}, \, \phi_{mm} \in C^\infty \text{ for } m \geq 2, \\
     &\text{assuming: }&& \ell' > 0, \ \kappa > 0, \ B_0 > 0, \ \phi_{00} = \phi_{10} = \phi_{11} = 0, \\
     &\text{solve: }&& \ell' \left(m^2 + \pd{^2}{\theta^2}\right) \phi_{m} = - (\nabla \cdot \tilde{\Bbf})_{m-2} -(\Delta \phi_{<m})_{m-2}, \\
     &\text{output:} && \text{potential }\phi \text{, magnetic field } \Bbf = \nabla \phi + \tilde{\Bbf}.
\end{aligned}
\end{empheq}

\subsection{Straight Field-Line Coordinates: Leading Order}
\label{subsec:Magnetic-Coordinates}
Given a solution magnetic field from Box \ref{eq:nae-box}, we consider the problem of finding straight field-line magnetic coordinates. 
We assume that the magnetic field is locally elliptic about the axis and the rotation number is irrational, so that the leading-order behavior is rotation about the magnetic axis.
This means that both hyperbolic orbits (x-points) and on-axis resonant perturbations are excluded from this work.
In the language of Hamiltonian normal forms, the \edit{leading order Hamiltonian is a Harmonic oscillator}{leading order field-line dynamics are conjugate to a non-resonant harmonic oscillator}; see \citet{burby_integrability_2021} and \citet{duignan_normal_2021} for a more rigorous derivation of magnetic coordinates in the near-axis expansion.
We note that our process of finding coordinates is formal: we make no claims that this problem converges in the limit. 
However, in Sec.~\ref{sec:examples}, we find that this procedure appears to converge well numerically.

To find magnetic coordinates, we attempt to build a conjugacy between magnetic field-line dynamics $\dot{\rbf} = (B^s)^{-1} \Bbf(\rbf)$ and straight field-line dynamics $\dot{\xibf} = (- \iota(\psi) \eta, \, \iota(\psi) \xi, \,  1)$, where $\xibf = (\xi, \, \eta, \, s) \in \Rbb^2 \times \Tbb$ are Cartesian coordinates, $\psi=\xi^2 + \eta^2$ is a flux-like coordinate, and $\iota$ is the rotational transform.
To make the connection with straight field-line coordinates precise, consider the transformation to polar coordinates $(\xi, \eta) \to \sqrt{\psi}(\cos \gamma, \sin \gamma)$. 
Then, the field-line is traced by $(\dot{\psi},\, \dot{\gamma},\, \dot{s}) = (0, \, \iota(\psi), \, 1)$, i.e.~magnetic field lines are straight with slope $\iota$.
However, we use the Cartesian version of magnetic coordinates because it removes the coordinate singularity associated with polar coordinates, simplifying the following steps.

\begin{figure}
    \centering
    \includegraphics[width=1.0\linewidth]{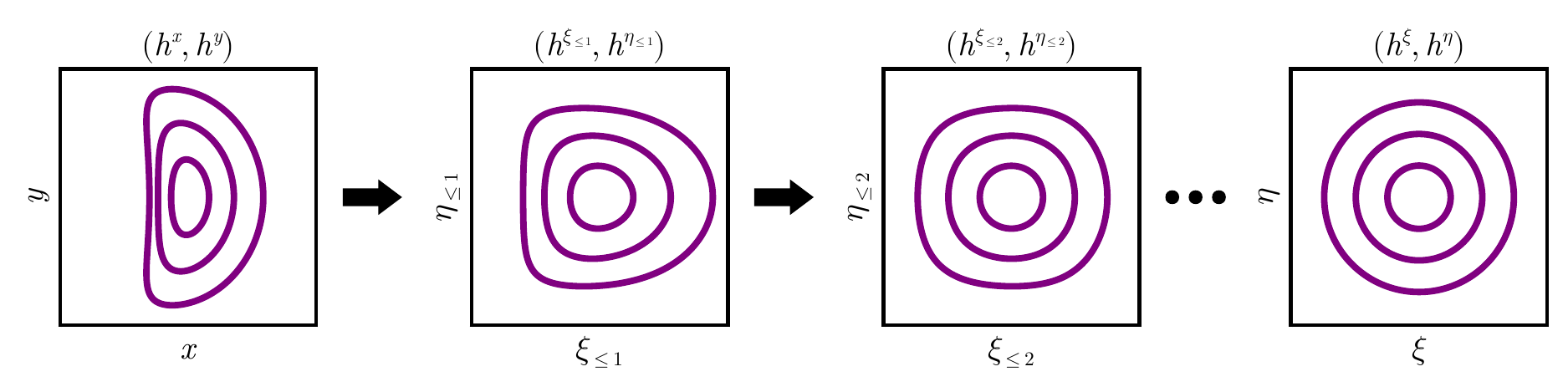}
    \caption{A schematic of the process of finding straight field-line coordinates. 
    On the left, we plot the surfaces of the magnetic field $(h^x,h^y)$ on a cross-section for fixed $s$. 
    Moving one plot to the right, the leading correction transforms to a coordinate frame where the main elliptic component is eliminated.
    Going one further, the next correction accounts for the most prominent triangularity.
    This process continues until, in $(\xi,\eta)$ coordinates, the magnetic surfaces are nested circles. }
    \label{fig:HamNoFo}
\end{figure}
There are two main steps to our process of finding magnetic coordinates: the leading-order problem and the higher-order problems (see Fig.~\ref{fig:HamNoFo} for a sketch of the process).
If we use the notation $\tilde{\xbf} = (x,\, y)$ and $\tilde{\xibf} = (\xi, \, \eta)$ for the out-of-plane coordinates, the leading order transformation takes the form
\begin{equation}
\label{eq:G0}
    \tilde{\xibf}_1 = \pd{\tilde{\xibf}_1}{\tilde{\xbf}}\tilde{\xbf} = G_0 \tilde{\xbf} = \begin{pmatrix}
        \xi_{10} x + \xi_{11} y \\
        \eta_{10} x + \eta_{11} y
    \end{pmatrix}.
\end{equation}
We will find that the problem for $G_0$ is an eigenvalue problem for the on-axis rotation number, as is typical for the linearized dynamics about a fixed point.
In the following section, we will discuss the inductive step to higher orders.

To begin, consider the contravariant form of the Cartesian near-axis magnetic field
\begin{equation*}
\frac{1}{B^s} \bm B = \frac{1}{B^s} \dd{\rbf}{\xbf} \begin{pmatrix}
    B^x \\ B^y \\ B^s
\end{pmatrix} = \frac{B^x}{B^s} \pd{\rbf}{x} + \frac{B^y}{B^s} \pd{\rbf}{y} + \pd{\rbf}{s}.
\end{equation*}
We would like to equate this to the straight field-line dynamics as
\begin{equation*}
    \frac{1}{B^s} \Bbf = \dd{\rbf}{\xibf} \begin{pmatrix}
        - \iota(\psi) \eta \\ \iota(\psi) \xi \\ 1
    \end{pmatrix},
\end{equation*}
where $\iota$ depends smoothly upon the radial label as
\begin{equation*}
    \iota = \iota_0 + \iota_2 \psi + \iota_4\psi^2 + \dots,
\end{equation*}
where we emphasize $\iota_\mu = 0$ for odd $\mu$. 
Multiplying both sides by $\df \xibf / \df \rbf$, we find the Floquet conjugacy problem
\begin{equation}
\label{eq:field-conjugacy}
    \pd{\tilde{\xibf}}{s} = -G(\xibf) \tilde{\hbf}^{\xbf} +\iota(\psi(\xibf)) J \tilde{\xibf},
\end{equation}
where
\begin{equation*}
    G = \pd{\tilde{\xibf}}{\tilde{\xbf}}, \qquad \tilde{\hbf}^{\xbf} = \begin{pmatrix}
        h^x \\ h^y
    \end{pmatrix} = \frac{1}{B^s} \begin{pmatrix}
        B^x \\ B^y
    \end{pmatrix}, \qquad J =  \begin{pmatrix}
        0 & -1 \\ 1 & 0
    \end{pmatrix},
\end{equation*}
where $J$ is known as the symplectic matrix.
Our problem is to solve \eqref{eq:field-conjugacy} for $\xi(x,y)$, $\eta(x,y)$, and $\iota(\psi)$. 

To find the leading-order problem for \eqref{eq:field-conjugacy}, we note the magnetic field $(h^x,h^y)$ is linear at leading order:
\begin{equation*}
    \tilde{\hbf}^{\xbf} = \begin{pmatrix}
        h^x_{10} x + h^x_{11} y \\
        h^y_{10} x + h^y_{11} y
    \end{pmatrix} + \Ocal(\psi) = H_0(s) \tilde{\xbf} + \Ocal(\psi).
\end{equation*}
Substituting this, Equation \eqref{eq:G0}, and $\iota = \iota_0 + \Ocal(\psi)$ into \eqref{eq:field-conjugacy}, we have
\begin{equation}
\label{eq:leading-conjugacy}
    \pd{G_0}{s} + G_0 H_0 = \iota_0 J G_0.
\end{equation}
The leading order problem \eqref{eq:leading-conjugacy} is a Floquet eigenvalue problem for the linearized field-line dynamics about the magnetic axis.
Assuming that the near-axis expansion is elliptic at leading order, the value of $\iota_0$ is real.
Otherwise, $\iota_0$ is not real, meaning ellipticity can be numerically verified for a given input.

There are many equivalent solutions to \eqref{eq:leading-conjugacy}, owing to the symmetries that if $(G_0, \iota_0)$ satisfies \eqref{eq:leading-conjugacy}, then the following are also solutions:
\begin{align}
\label{eq:iota-transformations}
    (R(ns) G_0, n + \iota_0), \qquad
    (J G_0, \iota_0), \qquad
    \left(\begin{pmatrix}
        0 & 1 \\
        1 & 0
    \end{pmatrix} G_0, -\iota_0\right),
\end{align}
where $R$ is a rotation matrix
\begin{equation*}
    R(\theta) = \begin{pmatrix}
        \cos\theta & - \sin\theta \\ \sin\theta & \cos\theta
    \end{pmatrix}.
\end{equation*}
The question of which solution to choose is then a question of practicalities.
\edit{Typically, $\iota$ is chosen so that it agrees with the winding number of the magnetic field about the magnetic axis in real space.
However, o}{For instance, one could choose the rotational transform corresponding to a non-twisting right-handed coordinate frame $\xibf_{1}$. 
Here, ``non-twisting'' means that closed coordinate lines near the axis, implicitly defined as curves $\bm r(x,y,s)$ where $\tilde{\xibf}(x,y,s) \neq 0$ is held constant, can be continuously deformed to a point on $\Rbb^3 \backslash \bm r_0$.
In other words, the coordinates lines do not link with the axis.
In this frame, $\iota_0$ agrees with the intuitive definition of the rotational transform as the limiting ratio of poloidal turns divided by toroidal turns of fieldlines about the axis. O}ther choices may have other benefits, e.g.~there may be an eigenfunction that $G_0$ behaves best numerically. In this paper, we opt for an option that is easy to implement: we take the real solution where $\iota_0$ has the smallest magnitude and $\det G_0$ is positive. 
From here, other equivalent coordinates can easily be found by applying the transformations in \eqref{eq:iota-transformations}.

For the scaling of $G_0$, we choose $\psi$ to be the actual magnetic flux at leading order. 
The formula for the flux is
\begin{equation*}
    \psi = \int_{\rbf(D_\psi, 0)} \Bbf \cdot \tbf \dif A = \int_{\rbf(D_\psi, 0)} B_s \dif A(\rbf),
\end{equation*}
where $D_\psi = \{(\xi,\eta) \mid \xi^2 + \eta^2 < \psi\}$ and
\begin{equation*}
    B_s = g_{sj} B^j = \pd{\phi}{s} + B_0 \ell'.
\end{equation*}
Pulling this back to the $\xibf$ frame, we have
\begin{equation}
\label{eq:flux}
    \psi = \int_{D_\psi}B_s(\xbf(\xibf, s), 0) (\det G)^{-1} \dif A(\xibf).
\end{equation}
Both $B_s$ and $G_0$ are constant in $\xibf$ to leading order, so \eqref{eq:flux} at leading order becomes 
\begin{equation*}
    \int_{D_\psi}B_s(\xbf(\xibf, s), 0) (\det G)^{-1} \dif A(\xibf) = \frac{(B_{s})_0}{\det G_0} \pi \psi + \Ocal(\psi^2).
\end{equation*}
Setting this equal to $\psi$, we find
\begin{equation*}
    \det G_0 = \pi (B_s)_0,
\end{equation*}
where we note that this is only possible when $(B_s)_0$ is chosen to be positive.

\subsection{Straight Field-Line Coordinates: Higher Order}
Now that we have the leading-order behavior, we iterate to go to higher order.
To do so, first define the near-axis expansion of $\tilde{\xibf}$ near the axis as
\begin{equation*}
    \tilde{\xibf} = \sum_{\mu=1}^\infty \tilde{\xibf}_\mu(\tilde{\xbf}, s), \qquad \tilde{\xibf}_\mu = \sum_{\nu = 0}^\mu \tilde{\xibf}_{\mu\nu}(s) x^{\mu-\nu} y^\nu,
\end{equation*}
and define the partial sums as
\begin{equation*}
    \tilde{\xibf}_{\leq \mu} = \sum_{\mu' = 1}^{\mu}\tilde{\xibf}_{\mu'},
\end{equation*}
where at leading order $\tilde{\xibf}_{\leq 1} = \tilde{\xibf}_1 = G_0 \tilde{\xbf}.$

At each order in the iteration, we consider the update to be a function of the previous coordinates, i.e.
\begin{align*}
    \tilde{\xibf}_{\leq \mu+1}(\xi_{\leq \mu}, \eta_{\leq \mu}, s_{\leq \mu}) &= \tilde{\xibf}_{\leq \mu} + \sum_{\nu=0}^\mu \tilde{\xibf}_{\mu \nu}(s_{\leq \mu}) \xi_{\leq \mu}^{\mu-\nu} \eta_{\leq \mu}^\nu, \\
    &= \tilde{\xibf}_{\leq \mu} + \sum_{\nu=0}^\mu \tilde{\xibf}_{\mu \nu} \xi_{1}^{\mu-\nu} \eta_{1}^\nu + \Ocal(\psi^{(\mu+2)/2}).
\end{align*}
We explicitly write the transformed toroidal coordinate $s_{\leq \mu}=s$ so that it is clear that $\pd{}{s}$ and $\pd{}{s_{\leq \mu}}$ are different operators. 
The purpose of performing the update in this way is primarily to make the update step as clear as possible. 

To wit, the magnetic field in the new frame satisfies
\begin{equation}
\label{eq:B-equality}
    \frac{1}{B^s} \Bbf = \dd{\rbf}{\xibf_{\leq \mu}} \begin{pmatrix}
        h^{\xi_{\leq \mu}} \\ h^{\eta_{\leq \mu}} \\ 1
    \end{pmatrix} = \dd{\rbf}{\xibf} \begin{pmatrix}
        -\iota \eta \\ \iota \xi \\ 1
    \end{pmatrix}.
\end{equation}
We can use the first equality to write
\begin{align}
\nonumber
    \tilde{\hbf}^{\xibf_{\leq \mu}}(\xibf_{\leq \mu}) &= \begin{pmatrix}
        h^{\xi_{\leq \mu}} (\xibf_{\leq \mu}) \\
        h^{\eta_{\leq \mu}} (\xibf_{\leq \mu}) 
    \end{pmatrix} \\
    &= \pd{\tilde{\xibf}_{\leq \mu}}{\tilde{\xbf}} \tilde{\hbf}^{\xbf}(\xbf(\xibf_{\leq \mu}))  + \pd{\tilde{\xibf}_{\leq \mu}}{s}, \\
\label{eq:transformed-field}
    &= [\iota(\psi) J \tilde{\xibf}]_{\leq \mu} + \left[ \pd{\tilde{\xibf}_{\leq \mu}}{\tilde{\xbf}}\tilde{\hbf}^{\xbf}(\xbf(\xibf_{\leq \mu}))\right]_{>\mu},
\end{align}
where we use the notation $f_{>\mu} = \sum_{\nu > \mu} f_{\nu}$. To get from the second to the third line, we have used the inductive assumption that we $\xibf_{\leq \mu}$ matches $\xibf$ up to order $\mu$, meaning that $\xibf_{\leq \mu}$ is a straight field-line coordinate system up to order $\mu$. 
In this way, we will find that the update residual depends neatly upon the transformed magnetic field.

It is worth noting that \eqref{eq:transformed-field} has two new operations that have not been introduced so far.
The first is that we are computing $\xbf(\xibf_{\leq \mu})$, i.e.~we are inverting the coordinate transformation.
The second is that we are composing functions with this inversion as $\tilde{\hbf}^{\xbf}(\xbf(\xibf_{\leq \mu}))$.
So, for any function $f(\xbf)$ and coordinate transformation $\xibf(\xbf)$, we can compute the equivalent function in indirect coordinates as $f(\xibf)$ using these two steps.
Moreover, the inverse transformation $\xbf(\xibf)$ can be precomputed for all transformations one wishes to perform of this type.
This gives a framework to move back and forth between direct and indirect near-axis formalisms to high order.
For more details on how these transformations are computed, see Appendices \ref{subsec:composition} and \ref{subsec:series-inversion}

Given the residual field \eqref{eq:transformed-field}, we can use the second equality in \eqref{eq:B-equality} to find the updated equation
\begin{equation}
\label{eq:arbitrary-conjugacy}
    \pd{\tilde{\xibf}}{s_{\leq \mu}} = - \pd{\tilde{\xibf}}{\tilde{\xibf}_{\leq \mu}} \tilde{\hbf}^{\xibf_{\leq \mu}}(\xibf_{\leq \mu}) + \iota(\psi) J \tilde{\xibf}.
\end{equation}
Note that this has the exact same form as \eqref{eq:field-conjugacy}, except we have shifted the underlying coordinates.
Substituting $\tilde{\xibf}_{\leq \mu+1} = \tilde{\xibf}_{\leq \mu} + \tilde{\xibf}_\mu$ into \eqref{eq:arbitrary-conjugacy}, we obtain
\begin{align}
\label{eq:eigenvalue-update}
    \pd{\tilde{\xibf}_{\mu+1}}{s_{\leq \mu}} + \iota_0 \pd{\tilde{\xibf}_{\mu+1}}{\tilde{\xibf}_{\leq \mu}} J \tilde{\xibf}_{\leq \mu} - \iota_0 J \tilde{\xibf}_{\mu+1} &= - \tilde{\hbf}^{\xibf_{\leq \mu}}_{\mu+1} + \iota_\mu \psi^{\mu/2}J \tilde{\xibf}_{\leq \mu}, \\
\nonumber
    &= \Fbf_{\mu+1}.
\end{align}

Because the leading order problem \eqref{eq:leading-conjugacy} is an eigenvalue problem, \eqref{eq:eigenvalue-update} has the form of a higher-order correction to the eigenvalue and eigenfunction.
To see what we might expect, consider the eigenvalue problem $K\bm y=\lambda M \bm y$ where each term is expanded in a small parameter, e.g.~$K=K_0 + K_1 \epsilon + K_2 \epsilon^2 + \dots$ for small $\epsilon$. 
The analogous update equation would be
\begin{equation*}
    (K_0 - \lambda_0 B_0)\bm{y}_{\mu} = \lambda_\mu M_0 \bm{y}_0 + \bm{R}_\mu,    
\end{equation*}
where $\bm{R}_\mu$ contains all of the residual terms.
Assume for simplicity that $\lambda_0$ is an isolated eigenvalue.
Then, there is a single secular term, which can be identified by taking the inner product of the above expression with the leading left eigenvector $\bm{z}_0$ to give $\lambda_\mu = -(\bm{z}_0^T \bm{R}_\mu)/(\bm{z}_0^T M_0 \bm{y}_0)$.
Once this is satisfied, the equation can be solved for $\bm{y}_\mu$, where we typically choose the free component in $\bm{y}_0$ so that the norm is constant. 

To perform the same steps on \eqref{eq:eigenvalue-update}, we first diagonalize the left-hand-side operator by converting to polar coordinates $(\xi_{\leq \mu}, \eta_{\leq \mu}) = R_{\leq \mu} (\cos\Theta_{\leq \mu}, \sin \Theta_{\leq \mu})$ as
\begin{equation*}
    \pd{\tilde{\xibf}_{\mu+1}}{s_{\leq \mu}} + \iota_0 \pd{\tilde{\xibf}_{\mu+1}}{\Theta_{\leq \mu}} - \iota_0 J \tilde{\xibf}_{\mu+1} = \Fbf_{\mu+1},
\end{equation*}
where
\begin{equation*}
    \tilde{\xibf}_{\mu+1} = (R_{\leq \mu})^{\mu+1}  \sum_{n=0}^{\mu} \sum_{\ell\in\Zbb} \tilde{\xibf}_{\mu+1,n\ell} e^{(2n-\mu)\I \theta+ \I \ell s_{\leq \mu}}.
\end{equation*}
After substitution, we find that
\begin{equation*}
    \left((\iota_0 (2n - \mu - 1) + \ell) \I I - \iota_0 J \right)\tilde{\xibf}_{\mu+1,n\ell} = \Fbf_{\mu+1,n\ell},
\end{equation*}
where the $2\times 2$ left-hand-side matrix has the eigenvalues $\lambda_{\mu+1,n\ell 0} = \iota_0 \I (2n - \mu) + \I \ell $ and $\lambda_{\mu+1,n \ell 1} = \iota_0 \I (2n - \mu - 2) + \I \ell$ with corresponding right-eigenvectors $(\vbf_0, \vbf_1) = ([-\I/\sqrt{2}, 1/\sqrt{2}], [\I/\sqrt{2},1/ \sqrt{2}])$. 
So, the resulting updates in the coefficients are
\begin{equation}
\label{eq:xi-coefficient-formula}
    \tilde{\xibf}_{\mu+1,n\ell} = \sum_{k\in\{0,1\}} \frac{1}{\I(\iota_0 (2n - \mu - 2k) + \ell)} \ip{\overline{\vbf}_{k}}{\Fbf_{\mu+1,n\ell}} \vbf_{k},
\end{equation}
where $\overline{\vbf}_k$ are the corresponding left eigenvectors due to $(\iota_0(2n-\mu-1)+\ell)\I I - \iota_0 J$ being skew-adjoint.

There are two cases where the update \eqref{eq:xi-coefficient-formula} fails.
The first case is when $2n - \mu - 2k = 0$ and $\ell = 0$, occurring only when $\mu$ is even. 
This is the standard secularity that indicates that $\iota_\mu$ must be updated, giving the condition that $\ip{\overline{\vbf}_{0}}{\Fbf_{\mu+1,\mu/2,0}} = \ip{\overline{\vbf}_{1}}{\Fbf_{\mu+1,\mu/2+1,0}} = 0$ for single-valued solutions, where we note that these formulas are equivalent for real magnetic fields. 
The resulting formula is
\begin{equation*}
    \iota_\mu = h^{\eta_{\leq \mu}}_{\mu+1,\mu/2,0} - \I h^{\xi_{\leq \mu}}_{\mu+1,\mu/2,0}.
\end{equation*}
The second case of failure is when $\iota_0$ is rational, as then there are other values of $(\mu,n,k,\ell)$ such that the formula \eqref{eq:xi-coefficient-formula} is singular.
This is attributed to the expanding number of $\theta$ modes at each order, where higher-order poloidal perturbations resonate with the axis.
To avoid this, extra resonant terms in the higher-order magnetic field must be introduced to avoid secularity \edit{}{\citep{duignan_normal_2021}} (equivalently, this requires adding terms to the Hamiltonian normal form). 
Here, we assume that $\iota_0$ is irrational so that the iteration is well-defined.

We note that there is still one undetermined part of the problem: what value to choose for $\langle{\overline{\vbf}_{0}},{\tilde{\xibf}_{\mu+1,\mu/2,0}}\rangle$ and its complex conjugate.
Because this is arbitrary, we currently set this coefficient to $0$.
However, other options could be to choose this value to improve the radius of convergence or to match the flux \eqref{eq:flux}. 
In summary, the straight field-line coordinate process is:
\begin{empheq}[box={\fboxsep=8pt\fbox}]{equation}
\label{eq:fieldline-box}
\begin{aligned}
    &\text{input:}&& \text{divergence-free magnetic field }\Bbf, \\
    &\text{assuming:} && \text{elliptic on axis with irrational }\iota_0, \\ 
\\ 
    &\text{leading eigenvalue problem:}&& \pd{G_0}{s} + G_0 H_0 = \iota_0 J G_0,\\ 
    &&& \det G_0 = \pi (B_s)_0, \qquad \abs{\iota_0} \text{ minimized}, \\ 
    & \text{where:} && \tilde{\xibf} = (\xi,\, \eta)^T = G_0(s) \tilde{\xbf} + \Ocal(\psi), \\
    &&& \tilde{\hbf}^{\xbf}(\xbf) = \frac{1}{B^s}\begin{pmatrix}B^x \\ B^y\end{pmatrix} = H_0(s) \tilde{\xbf} + \Ocal(\psi), \\
\\
     &\text{higher order linear solve:}&& \pd{\tilde{\xibf}_{\mu+1}}{s_{\leq \mu}} + \iota_0 \pd{\tilde{\xibf}_{\mu+1}}{\tilde{\xibf}_{\leq \mu}} J \tilde{\xibf}_{\leq \mu} - \iota_0 J \tilde{\xibf}_{\mu+1} = \Fbf_{\mu+1}, \\
     &&& \iota_\mu = h^{\eta_{\leq \mu}}_{\mu+1,\mu/2,0} - \I h^{\xi_{\leq \mu}}_{\mu+1,\mu/2,0},  \qquad (\mu \text{ even})\\
     &\text{where:}&& \Fbf_{\mu+1} = - \tilde{\hbf}^{\xibf_{\leq \mu}}_{\mu+1} + \iota_\mu \psi^{\mu/2}J \tilde{\xibf}_{\leq \mu}, \\
     &&& \tilde{\hbf}^{\xibf_{\leq \mu}}(\xibf_{\leq \mu}) = \pd{\tilde{\xibf}_{\leq \mu}}{\tilde{\xbf}} \tilde{\hbf}^{\xbf}(\xbf(\xibf_{\leq \mu})), \\
\\
    & \text{output:} && \text{rotational transform }\iota, \\
    &&&                 \text{straight field-line coordinates } \xibf.
\end{aligned}
\end{empheq}

\section{Ill-Posedness and Regularization}
\label{sec:ill-posedness}
In this section, we describe how the near-axis problem is ill-posed (\S\ref{subsec:ill-posed}) and how we can regularize the problem (\S\ref{subsec:regularization}). 
In \S\ref{subsec:reg-theorems}, we state how the near-axis expansion of $\phi$ converges under suitable input assumptions (Thm.~\ref{thm:reg-theorem} and Cor.~\ref{cor:convergence}).
Most proofs can be found in Appendix \ref{app:proofs}, where the individual sections are referred to after each statement.

\subsection{Ill-Posedness}
\label{subsec:ill-posed}
We define a problem as ill-posed if it is not well-posed, where the standard definition of a well-posed problem is that
\begin{enumerate}
    \item The solution exists.
    \item The solution is unique.
    \item The solution is continuous in the initial data.
\end{enumerate}
We note that the interpretations of these statements depend what space we require the solution to belong to and over which topology continuity is described in.
For instance, it is straightforward to show existence and uniqueness in the sense of a formal power series:
\begin{proposition}
\label{prop:formal-solvability}
    Consider the near-axis problem in box \eqref{eq:nae-box}, with all inputs in $C^{\infty}$. Then, there exists a unique formal power series solution $\phi_n(s)$ at each order.
\end{proposition}
\begin{proof}
    Simply notice that the residual at each order is $C^\infty$ if every previous order is. Then, because the inverse of $\Delta_\perp$ of $C^{\infty}$ functions is $C^\infty$, we satisfy the Fredholm alternative, and we specify the null space the operator at each order, we have a unique solution.
\end{proof}
Note that this proposition says nothing about the convergence of the power series to a solution off-axis; it only shows that we can find the coefficients of the power series. 
So, formal existence does not necessarily imply good or consistent computational results. 

Another straightforward existence result for harmonic inputs is the following:
\begin{proposition}
\label{prop:harmonic}
    Let $\Bbf = \nabla \phi + \tilde{\Bbf}$ be a valid vacuum magnetic field on $\Omega_\sigma$ with a real-analytic axis $\rbf_0$ and $\sigma>0$. Then the near-axis expansion using the coefficients corresponding to $\phi$ converges uniformly on a smaller domain $\Omega_{\sigma'}$ with $0<\sigma'<\sigma$. 
\end{proposition}
\begin{proof}
    See \S\ref{subsec:prop-harmonic-proof}
\end{proof}
Proposition \ref{prop:harmonic} is a useful result because it says that vacuum fields can, in principle, be written as solutions to the infinite near-axis problem. 
However, this is a difficult theorem to use in practice, as it is difficult to verify \textit{a priori} whether the input data to the near-axis expansion agrees with a solution of Poisson's equation.

So, for a more computational approach, we must define normed spaces of inputs and outputs that agree with notions of convergence.
To intuit what the correct space may be, we observe that the radial direction behaves as a ``time-like'' variable, whereas the $\theta$ and $s$ behave more like spatial variables.
That is, the near-axis PDE can be thought of as propagating surface information off the axis.
This motivates a decision to separate our treatment of these coordinates.
Moreover, we desire convergence in a power series in the radial variable, so we choose to treat it in an analytic manner.
In contrast, \edit{functions in the angles $\theta$ and $s$ result from the solution of PDEs}{the coefficients $\phi_m(\theta,s)$ are obtained by solving a linear PDE at each order}, so we treat them in a Sobolev sense \edit{}{as a function of $\theta$ and $s$}.

To this end, let $\alpha = (\alpha_1,\alpha_2)$ be a multi-index of degree $2$ and $\abs{\alpha} = \sum_j \alpha_j$.
We define the $H^q$ Sobolev norm of functions $f : \Tbb^2 \to \Rbb$ as 
\begin{equation*}
    \norm{f}^2_{H^q} = \sum_{\abs{\alpha} \leq q} \norm{D_\alpha f}^2_{L^2},
\end{equation*}
where 
\begin{equation*}
    D_\alpha f = \pd{^{\abs{\alpha}}f}{\theta^{\alpha_1} \partial s^{\alpha_2}}, \qquad \norm{f}^2_{L^2} = \int_{\Tbb^2}\abs{f}^2 \dif \mu,
\end{equation*}
and $\mu$ is the Lebesgue measure on $\Tbb^2$. 
We additionally define the $C^q$ norm of a $q$-times differentiable function $f : \Tbb^2 \to \Rbb$ as
\begin{equation*}
    \norm{f}_{C^q} = \sum_{\abs{\alpha}\leq q} \sup_{(\theta,s)\in\Tbb^2} \abs{D_\alpha f(\theta,s)}.
\end{equation*}
Then, we define \edit{an appropriate space of solutions as}{the following convenient near-axis function space:}
\begin{definition}
    Let $\sigma>0$ and $W$ be a Banach space on functions on $\Tbb^2$. We define a function $f : \Omega_\sigma^0 \to \Rbb^d$ as $(\sigma,W)$-analytic if $f$ has a convergent near-axis expansion of the form
    \begin{equation}
    \label{eq:analytic-form}
        f(\rho,\theta,s) = \sum_{m=0}^\infty f_m(\theta, s)\rho^m, \qquad
        f_m(\theta,s) = \sum_{n=0}^m f_{mn}(s) e^{(2n-m)\I \theta},
    \end{equation}
    where the norm
    \begin{equation*}
        \norm{f}_{\sigma,W} = \sup_n \left(\sigma^n \norm{f_n}_{W}\right)
    \end{equation*}
    is bounded. 
    Here, convergence of the near-axis expansion is pointwise in $\rho$ and in norm in $(\theta,s)$, i.e.~for all $\rho < \sigma$
    \begin{equation*}
        \lim_{M\to \infty}\norm{f(\rho,\cdot,\cdot) - \sum_{m=0}^M f_m(\cdot,\cdot) \rho^m}_{W} \to 0.
    \end{equation*}
\end{definition}

\edit{If $f$ is $(\sigma,W)$-analytic, then for any $F\geq \norm{f}_{\sigma,W}$,}{Paralleling standard linear regularity theory \citep{evans2010}, we will consider near-axis solutions $\phi$ to belong to a Sobolev $(\sigma,H^q)$-analytic space, while near-axis PDE coefficients belong to differentiable $(\sigma,C^q)$-analytic spaces. 
Because the coefficients of Poisson's equation \eqref{eq:poisson} are functions of the metric, control on the $(\sigma,C^q)$-analytic norm of the coefficients in the Frenet-Serret coordinate system will be entirely determined by the differentiability class of the axis $\bm r_0 \in C^q(\Tbb)$.}

\edit{}{To build some intuition for $(\sigma,W)$-analytic functions, we turn to some straightforward facts about their convergence. Let $f$ be $(\sigma,W)$-analytic. Then,  for any $F\geq \norm{f}_{\sigma,W}$, the definition of the norm $\norm{\cdot}_{\sigma,W}$ tells us that} the coefficients are bounded as
\begin{equation*}
    \norm{f_m}_{W} \leq F \sigma^{-m}.
\end{equation*}
This means that surfaces of $f$ converge geometrically for $\rho^* < \sigma$ in $W$, i.e.
\begin{equation*}
    \norm{\left. f \right|_{\rho=\rho^*}}_{W} \leq F\sum_{m=0}^{\infty} \left(\frac{\rho^*}{\sigma}\right)^{m} = \frac{F}{1-\rho^*/\sigma}.
\end{equation*}
\edit{We will primarily consider the space of Sobolev $(\sigma,H^q)$-analytic functions.}{}
For a more practical statement of pointwise convergence, we have
\begin{proposition}
\label{prop:continuity}
    Let $f$ be $(\sigma,C^q)$-analytic for $q \geq 0$. Then $f$ is \edit{}{continuous and} $q$-times \edit{}{continuously} differentiable in $\Omega_\sigma^0$.
\end{proposition}
\begin{proof}
    See \S\ref{subsec:prop-continuity-proof}.
\end{proof}
\begin{corollary}
\label{cor:Hq-continuity}
    Let $f$ be $(\sigma,H^q)$-analytic for $q \geq 2$. Then $f$ is $(\sigma,C^{q-2})$\edit{}{-analytic, continuous,} and $(q-2)$-times \edit{}{continuously} differentiable in $\Omega_\sigma^0$.
\end{corollary}
\begin{proof}
    Because $H^q(\Tbb^2)$ is continuously embedded in $C^{q-2}(\Tbb^2)$, $f$ is also $(\sigma,C^{q-2})$-analytic. Apply Prop.~\ref{prop:continuity}.
\end{proof}
\edit{}{A direct consequence of the preceding statements is that both $(\sigma,C^0)$-analytic and $(\sigma,H^2)$-analytic functions are continuous in 3D, but $(\sigma,H^1)$-analytic functions need not be.}

Now, let's return to the question of ill-posedness. 
To define the norm of the input, let
\begin{equation}
\label{eq:phiIC}
    \phiic = \sum_{m=0}^\infty \rho^m \left(\phi_{m0} e^{-\I m \theta} + \phi_{mm} e^{\I m \theta} \right).
\end{equation} 
This allows us to naturally define the norm of the input functions $\phi_{m0}$ and $\phi_{mm}$ via a single $(\sigma,\edit{C^q}{H^q})$-analytic norm. 
Using this, we prove that the problem is ill-posed in the following sense:
\begin{theorem}
\label{thm:illposed}
    Let $\rbf_0, B_0, \phi_{m0}, \phi_{mm} \in C^{\infty}$ for $2 \leq m < \infty$ with $\ell',\kappa> 0$ and $q_0 \geq 0$. The near-axis solution $\phi$ of \eqref{eq:nae-box} is not continuous to perturbations $\delta \rbf_0, \delta B_0, \delta \phi_{m0}, \delta \phi_{mm} \in C^{\infty}$ under the $C^{4+q_0}(\Tbb)$ norm on $\rbf_0$, the $H^{1+q_0}(\Tbb)$ norm on $B_0$, the $(\sigma_0,H^{2+q_0})$-analytic norm on $\phiic$ with $\sigma_0>0$, and any $(\sigma,H^q)$-analytic norm on the output $\phi$ with $\sigma>0$ and $q\geq 2$, i.e.~the near-axis expansion is ill-posed. 
\end{theorem}
\begin{proof}
    See \S\ref{subsec:illposed-proof}
\end{proof}

In other words, Theorem \ref{thm:illposed} tells us that smooth bounded perturbations in the input lead to unbounded deviations in the solution, even if the problem is initially prepared to be convergent.
The characteristic form of the unbounded perturbations are high frequency in $s$, which grow exponentially off-axis according to their wavenumber.
When the near-axis expansion is discretized, this appears to be a poor condition number for the truncated problem.
This motivates us to introduce a term in the near-axis expansion that damps the behavior of the high-frequency modes.

Before we continue, we note there is a strong connection between the ill-posedness of the near-axis expansion and the ill-posedness of coil design.
It is typically the case that magnetic fields with large gradients are difficult to approximate using plasma coils far from the boundary \citep{kappel2024}.
This is because high frequencies in coil design decay quickly towards the surface of the plasma --- the opposite view of the problem of high frequencies growing outward from the axis.
The effect is that it is difficult to match high frequencies on the plasma boundary, and coil design codes also require some form of regularization \citep{landreman_improved_2017}.

\subsection{Regularization}
\label{subsec:regularization}
\edit{
In this section, we introduce a method to regularize the vacuum near-axis expansion. 
The fundamental idea is to dampen the growth of highly oscillatory modes in $s$ while maintaining the fidelity of low-frequency modes. 
To do so, we propose adding a regularizing term to the near-axis problem as}{
We have just seen in Theorem \ref{thm:illposed} that the near-axis expansion described in Box \eqref{eq:nae-box} is ill-posed.
Ill-posedness can potentially cause significant problems for numerical simulations, with the primary one being that discretization refinement --- both in toroidal resolution and in the order of expansion --- will not lead to convergence to a solution.
Instead, the deviation from the correct answer typically increases with refinement.
}

\edit{}{
In the proof of Theorem \ref{thm:illposed} (\S\ref{subsec:illposed-proof}), the problematic perturbations to the input have small amplitude and high poloidal wavenumber, leading to exponential growth off the axis with a rate proportional to the wavenumber.
Here, we would like to damp the problematic high-wavenumber modes while maintaining the fidelity of the low-wavenumber modes. 
We do so by adding a high-order differential operator in the toroidal and poloidal directions, as high-order derivatives affect high wavenumbers significantly more that low wavenumbers.
Specifically, we propose the following version of Poisson's equation \eqref{eq:poisson} with regularization
}
\begin{equation}
\label{eq:regularized-laplace}
    \Delta \phi + \frac{\rho}{\sqrt{g}}\Delta_{\perp}P \phi  =  - \nabla \cdot \tilde{\Bbf}.
\end{equation}
\edit{where $P$ satisfies the following:}{The new term contains the regularizing differential operator $\Delta_\perp P$ that is at least fourth order. 
The operator is composed of a product of the perpendicular Laplacian defined by
\begin{equation*}
    \Delta_{\perp} \phi = \frac{1}{\rho} \pd{}{\rho}\left( \rho \pd{\phi}{\rho}\right) + \frac{1}{\rho^2} \pd{^2\phi}{\theta^2},
\end{equation*}
and the regularizing differential operator $P$ satisfying the following hypotheses:}
\begin{hypotheses}
\label{hypotheses}
    Let $q\geq0$ and $D \geq 1$. We require $P$ be an order $2D$ \edit{}{differential} operator that satisfies the following:
    \begin{enumerate}
        \item $P$ takes the form
        \begin{equation}
        \label{eq:P-form}
            P = \sum_{m=0}^{2D}\sum_{n=0}^{m} P^{(mn)}(s) \pd{^{m+n}}{s^m \partial \theta^n},
        \end{equation}
        where $P^{(mn)}(s)\in C^{q}$.
        \item $P$ is strongly elliptic, i.e.~for some $C>0$ and for all $(x,y)\in\Rbb^2 \backslash \{0\}$
        \begin{equation*}
            \sum_{n=0}^{2D} P^{(2D,n)}(s) x^{2D-n} y^n \geq C (x^2 + y^2)^D.
        \end{equation*}
        \item $P$ is self-adjoint and semi positive-definite, i.e.~for nonzero $f,g \in H^{2D}(\Tbb^2)$
        \begin{equation*}
            \ip{f}{Pg} = \ip{Pf}{g}, \qquad \ip{f}{Pf} \geq 0, \qquad \ip{f}{g} := \int_{\Tbb^2}f(\theta,s) g(\theta,s) \dif \theta \df s
        \end{equation*}
    \end{enumerate}
\end{hypotheses}
We note that in hypothesis (i), the coefficients of $P$ do not depend on $\theta$.
This means poloidal Fourier modes ``block diagonalize'' $P$, i.e.~for $f(s) \in H^{q + 2D}$, we have $P(f(s) e^{\I m \theta}) = g(s) e^{\I m \theta}$ for some $g(s) \in H^{q}$.
This is sufficient for $P$ to map $(\sigma,H^{q+2D})$-analytic functions to $(\sigma,H^{q})$-analytic functions, while $\theta$ dependence would make the operator non-analytic.
It also means that $P$ commutes with $\Delta_{\perp}$, i.e.~$P\Delta_\perp = \Delta_\perp P$ on sufficiently regular functions.
Moreover, the hypotheses (ii) and (iii) are sufficient for $1+P$ to be invertible, which is necessary at each step of the iteration.

In practice, we use the operator
\begin{equation}
\label{eq:P-reg}
    P = \left(-\frac{1}{K^2}\left(\frac{L}{2\pi\ell'} \pd{}{s}\right)^2\right)^{D} + \left(- \frac{1}{K^2} \pd{^2}{\theta^2} \right)^{D}
\end{equation}
where $K > 0$ is the characteristic wavenumber of regularization, $2D\geq 2$ is the algebraic order of the frequency damping, \edit{}{and} $L = \oint \ell' \dif s$ is the length of the magnetic axis\edit{, and the perpendicular Laplacian is}{.}
\edit{\begin{equation*}
    \Delta_{\perp} \phi = \frac{1}{\rho} \pd{}{\rho}\left( \rho \pd{\phi}{\rho}\right) + \frac{1}{\rho^2} \pd{^2\phi}{\theta^2}.
\end{equation*}}{}
We have chosen the specific form of $P$ so that it is easy to implement numerically when the curve is in arclength coordinates (and therefore $\ell'$ is constant, satisfying the smoothness requirement). 
Indeed, any polynomial in the arclength derivative $(\ell')^{-1}\pd{}{s}$ and the poloidal derivative $\pd{}{\theta}$ is simply inverted in Fourier space.
Additionally, the parameter $K$ can be tuned from large (weak damping) to small (strong damping) to adjust the regularization strength.

With regularization, the new iterative form of the near-axis expansion \eqref{eq:regularized-laplace} is 
\begin{equation*}
    \left(m^2 + \pd{^2}{\theta^2}\right)(1 + P) \phi_m = -(\nabla \cdot \tilde{\Bbf} + \Delta \phi_{<m})_{m-2},
\end{equation*}
where we discuss the relevant regularity in Cor.~\ref{cor:convergence}.
If we use the regularization \eqref{eq:P-reg} and $s$ to be a scaled arclength coordinate so that $\ell' = L/2\pi$ is constant, the componentwise version of the iteration is
\begin{multline*}
    (m^2 - (2n-m)^2)\left(1 + \left(\frac{ k}{K}\right)^{2D} + \left(\frac{2n-m}{K}\right)^{2D}\right)\phi_{mnk} = \\
    -(\nabla \cdot \tilde{\Bbf} + \Delta \phi_{<m})_{m-2,n-1,k}
\end{multline*}
for $0<n<m$. 
We see that in the near-axis iteration, the regularization damps the high-order modes by dividing by high-order polynomials in the poloidal and toroidal wavenumbers. 

By construction, the regularization has another benefit: the full problem can be represented as the divergence of a perturbed magnetic field $\Bbf_K$.
If we let $G = \Diag(1, \rho^{-2}, 0)$, we have
\begin{equation}
\label{eq:BK}
    B_K^i = g^{ij} \pd{\phi}{q^j} + \tilde{B}^i + \frac{\rho}{\sqrt{g}} G^{ij} \pd{(P \phi)}{q^j} .
\end{equation}
The fact that there is still an underlying divergence-free field is important for the relation between the near-axis expansion and Hamiltonian mechanics. 
Therefore, the regularized near-axis expansion of the field $\Bbf_K$ can be expressed by the problem
\begin{equation}
\label{eq:fictitious-current}
     \nabla \times \Bbf_K = \bm J_K, \qquad \nabla \cdot \Bbf_K = 0, 
\end{equation}
where $\bm J_K$ is a fictitious regularizing current.
In contrast, the regularization means $\Bbf = \nabla \phi + \tilde{\Bbf}$ is not divergence-free.
To find the flux surfaces, we note that the procedure in Box \ref{eq:fieldline-box} in no way depends on the magnetic field being curl-free.
As such, we perform the same steps for finding flux coordinates for $\Bbf_K$ as with $\Bbf$. 

\subsection{Convergence of the Regularized Expansion}
\label{subsec:reg-theorems}
Consider a PDE of the form
\begin{equation}
\label{eq:arbitrary-PDE}
    \Delta_\perp (1+P)\phi = f + L^{(a)}\phi,
\end{equation}
where $P$ satisfies Hypotheses \ref{hypotheses}, $f$ is $(\sigma_0,H^{q})$-analytic for $q\geq 0$ and $\sigma_0>0$, and 
\begin{multline}
\label{eq:La}
    L^{(a)} = \ac{1}\phi + \pd{\ac{2}}{\rho} \pd{\phi}{\rho} + \frac{1}{\rho^2} \pd{\ac{2}}{\theta}\pd{\phi}{\theta} +  \ac{3} \pd{\phi}{\theta}  \\+ \ac{4} \pd{\phi}{s} + \ac{5} \pd{^2\phi}{\theta^2} + \ac{6} \pd{^2\phi}{\theta \partial s} + \ac{7} \pd{^2 \phi}{s^2},
\end{multline}
and each $a^{(j)}$ is a $(\Sigma, C^{q})$-analytic for $\Sigma > \sigma_0$. 
By a straightforward application of chain rule, we see that \eqref{eq:regularized-laplace} satisfies this form for $\bm r_0 \in H^{4+q}$ and $B_0 \in H^{1+q}$, where $\Sigma \leq \min \kappa^{-1}$ (see proof of Cor.~\ref{cor:convergence} in App.~\ref{subsec:convergence}). 
Equation \eqref{eq:arbitrary-PDE} also encompasses other coordinate and regularization assumptions of the near-axis expansion, including those not defined by Frenet-Serret coordinates.

The implicit near-axis iteration for \eqref{eq:arbitrary-PDE} can be written as
\begin{equation*}
    (\Delta_\perp(1+ P) \phi)_m = \left[f + L^{(a)}\right]_{m}.
\end{equation*}
The form of the iteration automatically respects the Fredholm condition, so it is solvable at each order (see the proof of Theorem \ref{thm:reg-theorem}).
So, to solve at each order of the iteration, we can invert to find that
\begin{equation*}
    \phi_m = \phiic_m + \left( (1+P)^{-1} \Delta_\perp^{+} \left[f + L^{(a)}\phi \right]\right)_{m},
\end{equation*}
where $\phiic$ is $(\sigma_0,H^{q+2})$-analytic of the form \eqref{eq:phiIC} and we define the pseudoinverse $\Delta^+_\perp$ as
\begin{equation}
\label{eq:Delta-plus}
    \Delta_\perp^{+} (\rho^m e^{\I (2n-m) \theta}) = \frac{1}{(m+2)^2 - (2n-m)^2} \rho^{m+2} e^{\I (2n-m) \theta}, \qquad  \text{ for }0\leq n \leq m.
\end{equation}
Given this iteration, we find the following theorem for its convergence:
\begin{theorem}
\label{thm:reg-theorem}
    Let $0<\sigma_0<\Sigma$, $q \geq 1$, and $D\geq 1$. Then, let $f$ be $(\sigma_0,H^q)$-analytic, $a^{(j)}$ for $1\leq j \leq 7$ be $(\Sigma, C^q)$-analytic, $\phiic$ as defined by \eqref{eq:phiIC} be $(\sigma_0,H^{q+2D})$-analytic, and $P$ satisfy Hypotheses \ref{hypotheses}. 
    Then, there is a unique $(\sigma,H^{q+2D})$-analytic near-axis solution $\phi$ of \eqref{eq:arbitrary-PDE} with initial data $\phiic$.
    Moreover, this solution operator is continuous in $a^{(j)}$ and satisfies the bound
    \begin{equation*}
        \norm{\phi}_{\sigma,H^{q+2D}} \leq C \norm{f}_{\sigma_0,H^q} + C \norm{\phiic}_{\sigma_0,H^{q+2D}}.
    \end{equation*}
\end{theorem}
\begin{proof}
    See Appendix \ref{sec:reg-proof}.
\end{proof}

This theorem can be translated to the Frenet-Serret problem:

\begin{corollary}
\label{cor:convergence}
    Let $\sigma>0$, $q \geq 1$, and $D\geq 1$. Then, let $\bm r_0 \in C^{4+q}$ in scaled arclength coordinates with $\ell' = L/2\pi > 0$ and $\kappa > 0$, $B_0 \in H^{1+q}$, $\phiic$ be $(\sigma_0,H^{q+2D})$-analytic, and $P$ be defined as in \eqref{eq:P-reg}. 
    Then, for some $\sigma>0$ satisfying $\sigma \leq \sigma_0$ and $\sigma < \min \kappa^{-1}$, the near-axis solution $\phi$ of \eqref{eq:regularized-laplace} is $(\sigma,H^{q+2D})$-analytic, continuous in $\bm r_0$, and satisfies
    \begin{equation*}
        \norm{\phi}_{\sigma,H^{q+2D}} \leq C \norm{B_0}_{H^{1+q}} + C \norm{\phiic}_{\sigma_0,H^{q+2D}}.
    \end{equation*}
    Furthermore, the associated divergence-free field $\Bbf_K$ is well-defined and is $(\sigma',H^q)$-analytic for all $\sigma'<\sigma$.
\end{corollary}
\begin{proof}
    See Appendix \ref{subsec:convergence}.
\end{proof}

There are two primary reasons why Theorem \ref{thm:reg-theorem} is useful for computation. 
First, it gives a guarantee that the norm of the inputs controls the size of the output.
For instance, in the context of stellarator optimization, if an appropriate Sobolev penalty is put on $\bm r_0$, $B_0$, and $\phiic$, then one can expect the output to be appropriately bounded.
Second, because the output is $(\sigma,H^{q+2D})$-analytic, it tells us that truncations of the near-axis expansion are approximations of the true solution. 
So, we can be more confident that finite asymptotic series are approximately correct, at least for the regularized problem.

It is worth noting that while Theorem \ref{thm:reg-theorem} tells us there is a solution to the regularized problem, it does not tell us that the solution solves the original problem.
To address this, we can develop an \textit{a posteriori} handle on the error.
\begin{proposition}
\label{prop:a-posteriori}
    Consider the hypotheses of Theorem \ref{thm:reg-theorem} with $q\geq 2$. Additionally, suppose $0<\sigma'<\sigma$, $L=\Delta_\perp - L^{(a)}$ is a second order negative-definite operator on $H^{2}_{0}(\Omega_{\sigma'}^0)$, and $\phi^{\sigma'} = \phi(\sigma',\theta,s)$.
    Then, the solution $\phi$ of the near-axis expansion is the unique solution in $H^{q+2D}(\Omega_{\sigma'}^0)$ of the boundary value problem 
    \begin{equation*}
        P \Delta_\perp\phi + L \phi = f, \qquad \left. \phi\right|_{\rho = \sigma'} = \phi^{\sigma'}.
    \end{equation*}
    Moreover, let $\tilde\phi$ be the solution to 
    \begin{equation*}
        L \tilde \phi = f, \qquad \left. \tilde \phi \right|_{\rho=\sigma'} = \phi^{\sigma'}.
    \end{equation*}
    Then,
    \begin{equation*}
        \norm{\phi - \tilde \phi}_{H^{q}(\Omega_{\sigma'}^0)} \leq C \norm{P \Delta_\perp \phi}_{H^{q-2}(\Omega_{\sigma'}^0)}.
    \end{equation*}
\end{proposition}
\begin{proof}
    See Appendix \ref{subsec:a-posteriori}.
\end{proof}
In other words, this tells us that given a solution to the regularized problem, we can bound the distance to a non-regularized boundary value problem via the norm of $P \Delta_\perp \phi$. 
This applies directly to the Frenet-Serret case because the Laplacian is negative-definite. 
So, given a solution, this gives us an estimate of the error.
As before, we can summarize the new problem (cf.~Box \eqref{eq:nae-box}):
\begin{empheq}[box={\fboxsep=6pt\fbox}]{equation}
\label{eq:reg-box}
\begin{aligned}
     &\text{input: }&& \text{axis } \rbf_0\in C^{4+q} \text{, on-axis field } B_0\in H^{1+q}, \\
    &&&                (\sigma_0,H^{q+2D})\text{-analytic higher moments } \phiic, \\
    &&&                 P \text{ satisfying Hypotheses \ref{hypotheses},} \\
     &\text{assuming: }&& \ell' > 0, \ \kappa > 0, \ B_0 > 0, \ q,D\geq 1\\
     &\text{solve: }&& \ell' \left(m^2 + \pd{^2}{\theta^2}\right) (1+P) \phi_{m} = - (\nabla \cdot \tilde{\Bbf})_m -(\Delta \phi_{<m})_m, \\
     &\text{output:} && (\sigma,H^q)\text{-analytic potential }\phi,\\
    &&&                 \text{magnetic field } B_K^i = g^{ij} \pd{\phi}{q^j} + \tilde{B}^i + \frac{\rho}{\sqrt{g}} G^{ij} \pd{(P \phi)}{q^j}.
\end{aligned}
\end{empheq}

\section{Numerical Method}
\label{sec:numerical-method}
To numerically solve the regularized near-axis expansion algorithm in \eqref{eq:reg-box}, we use a pseudospectral method. 
Pseudospectral methods use spectral representations of the solution for derivatives, while scalar multiplication and other operations occur on a set of collocation points.
The spectral form of the series is 
\begin{gather}
\label{eq:real-fourier-form}
    X_{\leq\Nrho}(\rho,\theta,s) = \sum_{m=0}^{\Nrho} \rho^m X_m(\theta, s), \qquad X_m(\theta, s) =  \sum_{n=0}^m  X_{mn}(s) \Fcal_{2n-m}(\theta) , \\
\nonumber
    X_{mn}(s) = \sum_{k = -\Ns}^{\Ns} X_{mnk} \Fcal_k(s)
\end{gather}
where $\Nrho$ and $\Ns$ are integers specifying the resolution of the series and
\begin{equation*}
    \Fcal_k(s) = \begin{cases} 
        \cos(k s), & k \geq 0, \\
        \sin(-k s), & k < 0.
    \end{cases}
\end{equation*}
Derivatives of the series are numerically evaluated by
\begin{align*}
    \pd{X_{\leq\Nrho}}{\rho} &= \sum_{m=0}^{\Nrho-1}  (m+1) X_{m}(\theta, s) \rho^{m}, \\
    \pd{X_{\leq\Nrho}}{\theta} &= \sum_{m=1}^{\Nrho} \sum_{n=0}^m -(2n-m) X_{mn} \rho^{m} \Fcal_{-(2n-m)}(\theta),\\
    \pd{X_{\leq\Nrho}}{s} &= \sum_{m=0}^{\Nrho} \sum_{n=0}^m \sum_{k=-\Ns}^{\Ns} -k X_{mnk} \rho^{m} \Fcal_{2n-m}(\theta) \Fcal_{-k}(s).
\end{align*}

For algebraic operations such as series multiplication, composition, and inversion, we discretize each $X_m$ on a grid $X_m(\theta_{mj}, s_\ell)$ where
\begin{equation}
\label{eq:collocation-nodes}
    \theta_{mj} = \edit{\frac{2 \pi j}{2m+1}}{\frac{\pi j}{m+1}}, \qquad s_\ell = \frac{2 \pi \ell}{M_s},
\end{equation}
where $0 \leq j \leq m$, $0 \leq \ell < M_s$, and $M_s \geq 2\Ns+1$ is the number of $s$-collocation points. 
Typically, we choose $M_s = 4\Ns + 3$ to oversample in $s$ by a factor of over $2$.
This choice anti-aliases the numerical method by removing high harmonics generated in the collocation space \citep{boyd_chebyshev_2001}. 
\edit{
The $\theta$-collocation in \eqref{eq:collocation-nodes} is specialized for the analytic form of the near-axis expansion.
In particular, because there are $m$ Fourier modes in $\theta$ at each order, we choose exactly the same number of collocation points as Fourier modes at each order. 
To see why it is necessary to discretize $\theta$ over the half circle instead of the full circle, choose $m=1$ and consider the alternative equispaced collocation nodes $\{\tilde{\theta}_{10},\tilde{\theta}_{11}\}=\{0,\pi\}$.
On these nodes, we see that $\Fcal_{-1}(\tilde{\theta}_{1n}) = \sin(\tilde{\theta}_{1n})=0$ for $n\in\{0,1\}$, so the transformation between collocation nodes and Fourier coefficients would be singular.
}{
We note that the $\theta$-collocation points $\theta_{mj}$ are spaced around the half circle instead of the full circle, owing to the fact that analytic coefficients satisfy the symmetry
\begin{equation}
\label{eq:Xm-symmetry}
    X_m(\theta+\pi,s) = (-1)^m X_m(\theta,s).
\end{equation}
For even $m$, \eqref{eq:Xm-symmetry} tells us that $X_m$ is periodic in $2\theta$, and the resulting transformation is the Discrete Fourier Transform (DFT) in that angular coordinate.
On the other hand, $X_m$ is anti-periodic in $\theta$ for odd $m$, and the collocation on the half-circle can interpreted as a symmetry reduction of the DFT on the full circle.
}

To transform between Fourier and spatial representations at each order $m$, let $\mathsf{X}^{\mathrm{c}}_{m} = [X_m(\theta_{mj},s_{\ell})] \in \Rbb^{(m+1) \times M_s}$ be the matrix of collocation values and $\mathsf{X}^{s}_m = [X_{mnk}] \in \Rbb^{(m+1)\times(2\Ns+1)}$ be the matrix of Fourier coefficients.
Then we define the transition matrices $[(\mathsf{F}_{\theta}^{m,m'})_{jn}] = [\Fcal_{2n-m}(\theta_{m'j})] \in \Rbb^{(m+1)\times(m'+1)}$ and $[(\mathsf{F}_s)_{k\ell}]=[\Fcal_{k}(s_\ell)] \in \Rbb^{(2N_s+1)\times M_s}$. 
The transformation from spectral coefficients to collocation nodes is expressed by
\begin{equation*}
    \mathsf{X}^c_m = (\mathsf{F}_{\theta}^{m,m})^T \mathsf{X}^s_m \mathsf{F}_s.
\end{equation*}
Similarly, the inverse transform happens via
\begin{equation*}
    \mathsf{X}^s_m = (\mathsf{F}_{\theta}^{m,m})^{-T} \mathsf{X}^c_m \mathsf{F}_s^+,
\end{equation*}
where the pseudo-inverse is $\mathsf{F}_s^+ = \mathsf{F}_s^T D_s^{-1}$ and $D_s = \mathsf{F}_s \mathsf{F}_s^T \in \Rbb^{(2N_s+1)\times(2N_s+1)}$ is diagonal with $(D_s)_{jj}=(2N_s+1)$ for $j=0$ and $(D_s)_{jj} = (2N_s+1)/2$ for $-N_s\leq j\leq N_s$, $j\neq 0$. 
This transformation is currently performed via full matrix-matrix multiplication, but it could be accelerated for large systems by the fast Fourier transform.

The final basic operation we use is to raise the order of the $\theta$-collocation. 
To see why this is necessary, consider the simple case of three monomial power series: $X = \rho^{m} X_m$, $Y = \rho^{m'} Y_{m'}$ and $Z = XY = \rho^{m+m'} Z_{m+m'} = \rho^{m+m'} X_m Y_{m'}$.
Then, the multiplication on the collocation nodes as
\begin{equation*}
    Z(\theta_{m+m',j}, s_\ell) = \rho^{m+m'} X(\theta_{m+m',j}, s_\ell) Y(\theta_{m+m',j}, s_j).
\end{equation*}
So, to obtain the correct collocation on $Z_{m+m'}$, we need to change the $\theta$-collocation on $X_{m}$ from $\theta_{mj}$ to $\theta_{m+m',j}$ and similarly for $Y_{m'}$.
To do this, we use the $\theta$-collocation matrices to find
\begin{equation*}
    \mathsf{Z}^c_{m+m'} = \left[(F_{\theta}^{m,m+m'})^T(F_{\theta}^{m,m})^{-T} \mathsf{X}^c_m\right] \odot \left[(F_{\theta}^{m',m+m'})^T(F_{\theta}^{m',m'})^{-T} \mathsf{Y}^c_m\right],
\end{equation*}
where $\odot$ is the Hadamard (element-wise) product and $\mathsf{Y}^c_{m'}$ and $\mathsf{Z}^c_{m+m'}$ are the collocation matrices of $Y_{m'}$ and $Z_{m+m'}$.
With this operation, the operations outlined in Appendix \ref{app:operations} can be performed on the collocated nodes.

We note that choosing the correct amount of modes in $s$ presents the most difficult numerical problem in computing the near-axis expansion. 
As the order increases, high-order residuals are increasingly nonlinear in the lower orders, causing a broadening of the spectrum in $s$.
If the inputs to the expansion do not have a sufficiently narrow bandwidth, this will result in broad higher-order residuals, particularly when finding flux coordinates.
Both the regularization and the anti-aliasing effects of choosing $F_s$ to be rectangular help alleviate the issue of broad bandwidth, but in practice we have found that it remains important to choose smooth inputs, especially for finding flux surfaces.

\section{Examples}
\label{sec:examples}
We now investigate the numerical convergence of the near-axis expansion to high orders.
Our focus is on characterizing the convergence of the input (Fig.~\ref{fig:phiic}), the convergence of the output magnetic field (Fig.~\ref{fig:Bfield}), the convergence of the magnetic surfaces (Figs.~\ref{fig:surfaces}, \ref{fig:surface-conv}), and the role of regularization (Fig.~\ref{fig:Bfield}, \ref{fig:finite-diff}).
Through our two examples --- the rotating ellipse and the precise QA equilibrium of Landreman-Paul \citep{landreman-paul} --- we find that the radius of convergence of every series is closely related to the distance from the magnetic axis to the coils $\sigma_{\mathrm{coil}}$.
This radius appears to limit the convergence of every other series of interest, including the magnetic surfaces that extend beyond this distance.

All computations in this section were performed on a personal laptop. 
The code used to perform the expansions can be found at the \texttt{StellaratorNearAxis.jl} package \citep{sna-jl}.

\subsection{Equilibrium Initialization}
\begin{figure}
    \centering
    \begin{subfigure}[t]{0.49\textwidth}
        \centering
        \includegraphics[width=\linewidth]{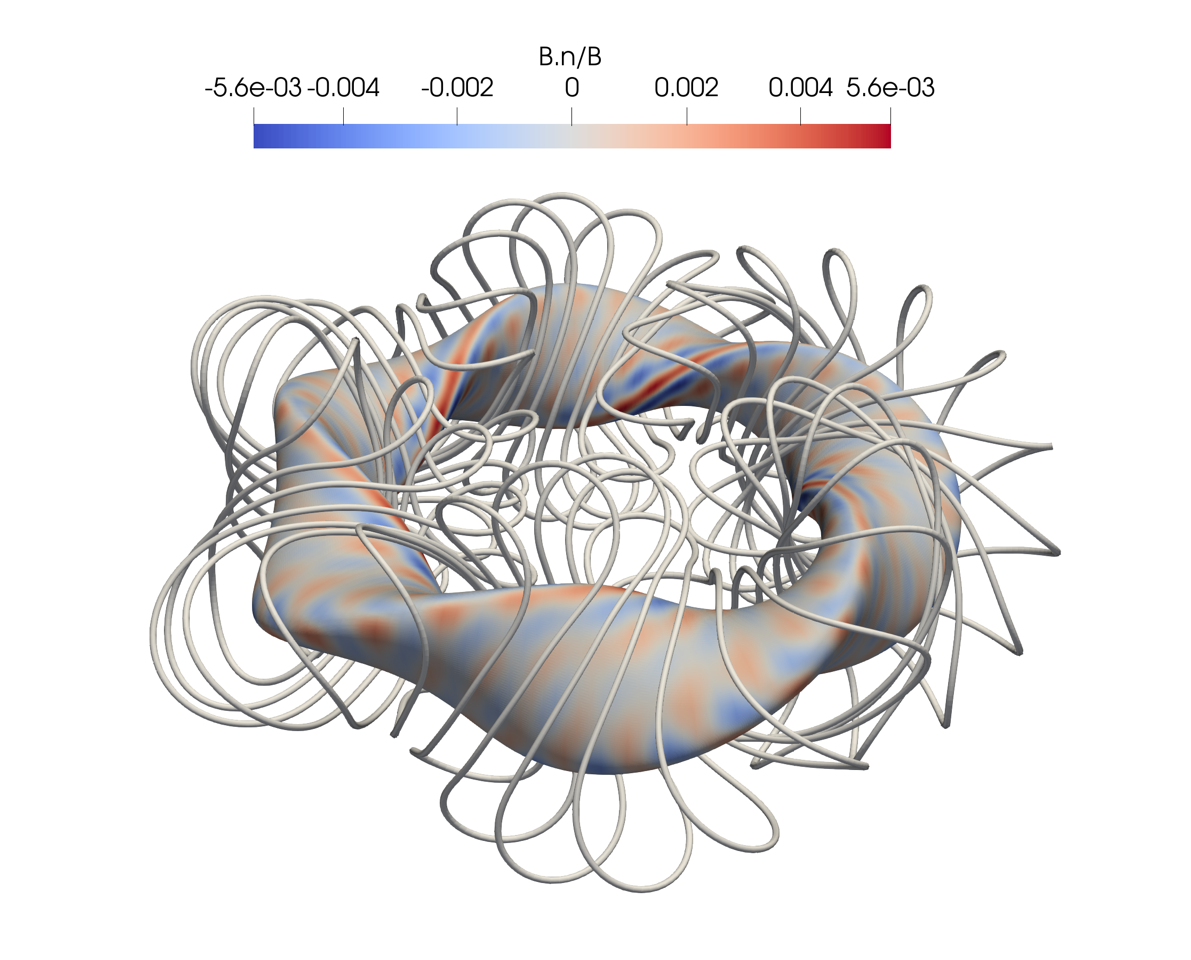}
        \caption{Rotating ellipse}
    \end{subfigure}
    \begin{subfigure}[t]{0.49\textwidth}
        \centering
        \includegraphics[width=\linewidth]{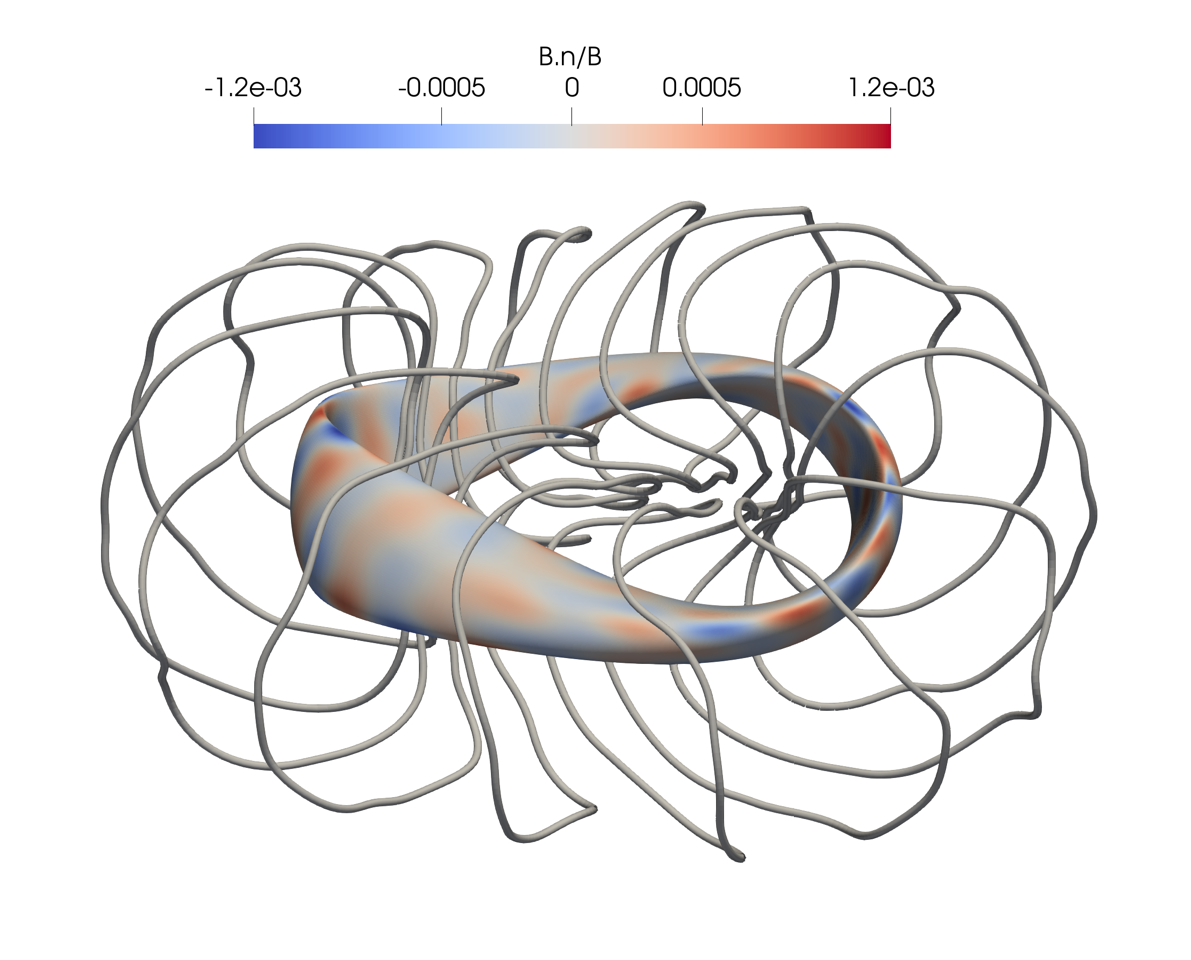}
        \caption{Landreman-Paul}
    \end{subfigure}
    \caption{Coil sets for the rotating ellipse and Landreman-Paul examples. The color indicates the normalized $\Bbf \cdot \bm N$ error on the outer closed flux surface.}
    \label{fig:coils}
\end{figure}

A major task in computing high-order near-axis expansions is choosing the input coefficients $\phiic$ in Box \ref{eq:reg-box}. 
While low-order expansions can often be expressed in physically intuitive variables parameterizing the rotation and stretching of elliptical magnetic surfaces, it is not as intuitive how to determine the high-order coefficients of $\phiic$.
In practice, the best option for finding equilibria would likely be via optimization.
However, for the purposes of demonstration, we initialize our inputs by a more direct method: via magnetic coils (see Fig.~\ref{fig:coils}).

The primary advantage of using coils for equilibrium initialization is accuracy.
The accuracy comes from the fact that the coil field can be expanded analytically about the axis, giving a direct input to the near-axis expansion. 
This circumvents the potentially error-prone problem of interpolating stellarator equilibria.
Coils also provide an accurate ground truth to compare our equilibrium to. 
In the following subsections, we use this to assess the accuracy of the near-axis expansion, both close to the axis and farther away.

The coil optimization method employed here follows the approach described in \citet{wechsung_precise_2022} and \citet{jorge_simplified_2024}
using the code SIMSOPT \citep{simsopt}.
Coils are modeled as single closed 3D filaments of current $\bm \Gamma^{(i)} : \Tbb \to \Rbb^3$. 
Each coil $i$ is modeled as a periodic function in Cartesian coordinates where
\begin{equation*}
    \bm \Gamma^{(j)}(s)=\bm c_{0}^{(j)}+\sum_{\ell=1}^{N_F}\left[\bm c_{\ell}^{(j)}\cos(\ell \theta)+\bm s_{\ell}^{(j)}\sin(\ell \theta)\right],
\end{equation*}
where each $\bm c_{\ell}^{(j)}, \bm s_{\ell}^{(j)} \in \Rbb^3$, yielding a total of $3\times(2N_F +1)$ degrees of freedom per coil. In this work, we used $N_F=12$, with 4 coils per half-field period for Landreman-Paul case and 8 coils per half-field period for the ellipse.
The degrees of freedom for the coil shapes are then
\begin{equation}
    \bm x_{\text{coils}}=[\bm c_{l}^{(j)},\bm s_{l}^{(j)}, I_j],
\end{equation}
with $I_j$ the current that goes through each coil.
We take advantage of stellarator and rotational symmetries to only optimize a set of $N_c$ coils per half field-period.
This leads to a total of $2 \times n_{\text{fp}} \times N_c$ modular coils where $n_{\text{fp}}$ is the number of toroidal field periods with $n_{\text{fp}} = 2$ for Landreman-Paul and $n_{\text{fp}} = 5$ for the rotating ellipse. 
The remaining coils are determined by symmetry.
The magnetic field $\mathbf B_{\text{ext}}$ of each coil is evaluated using the Biot-Savart law
\begin{equation}
\label{eq:Bcoil}
    \Bbf_{\mathrm{coil}}(\rbf) = \frac{\mu_0}{4\pi} \sum_{j=1}^{2 n_{\mathrm{fp}} N_c} \int_{0}^{2\pi} \frac{I_j \pd{\bm \Gamma_j}{s'} \times (\bm \Gamma_j(s') - \rbf)}{\abs{\bm \Gamma_j(s') - \rbf}^3} \dif s',
\end{equation}
where $s'$ parameterizes the coil curve.
Each coil is divided into 150 quadrature points, and the cost functions used to regularize the optimization problem use the minimum distance between two coils, the length of each coil, their curvature, and mean-squared curvature \citep[see][]{wechsung_precise_2022}.

Using the coil magnetic field \eqref{eq:Bcoil}, we find the magnetic axis $\rbf_0$ via a shooting method.
Then, using the near-axis coordinate representation of $\rbf$ in \eqref{eq:general-coord}, we expand the quadrature rule of \eqref{eq:Bcoil} using the operations in App.~\ref{app:operations} to find a near-axis expansion for the magnetic field $\Bbf(\rho,\theta,s)$.
Given the near-axis field, it is straightforward to compute $B_0$ by 
\begin{equation*}
    B_0(s) = \tbf(s) \cdot \Bbf(0,0,s),
\end{equation*}
and $\phi$ is found by a near-axis expansion of the path integral (note that $\phi=0$ on the axis)
\begin{equation*}
    \phi(\rho,\theta,s) = \int_{0}^\rho (\Bbf(\rho',\theta,s) - \tilde{\Bbf}(\rho',\theta,s)) \cdot (\cos(\theta) \nbf(s) + \sin(\theta) \bbf(s)) \dif \rho'.
\end{equation*}
Finally, the coefficients $\phi_{m0}$ and $\phi_{mm}$ of this are used as input for $\phiic$.
The input is computed to the $\Nrho=9$ orders in $\rho$ with $N_{s,\mathrm{RE}}=100$ and $N_{,s\mathrm{LP}}=50$ Fourier modes in $s$ for the rotating ellipse and Landreman-Paul respectively (see~\eqref{eq:real-fourier-form}), where we use the subscript `RE' for the rotating ellipse and `LP' for Landreman-Paul wherever necessary.

\begin{figure}
    \centering
    \includegraphics[width=0.6\linewidth]{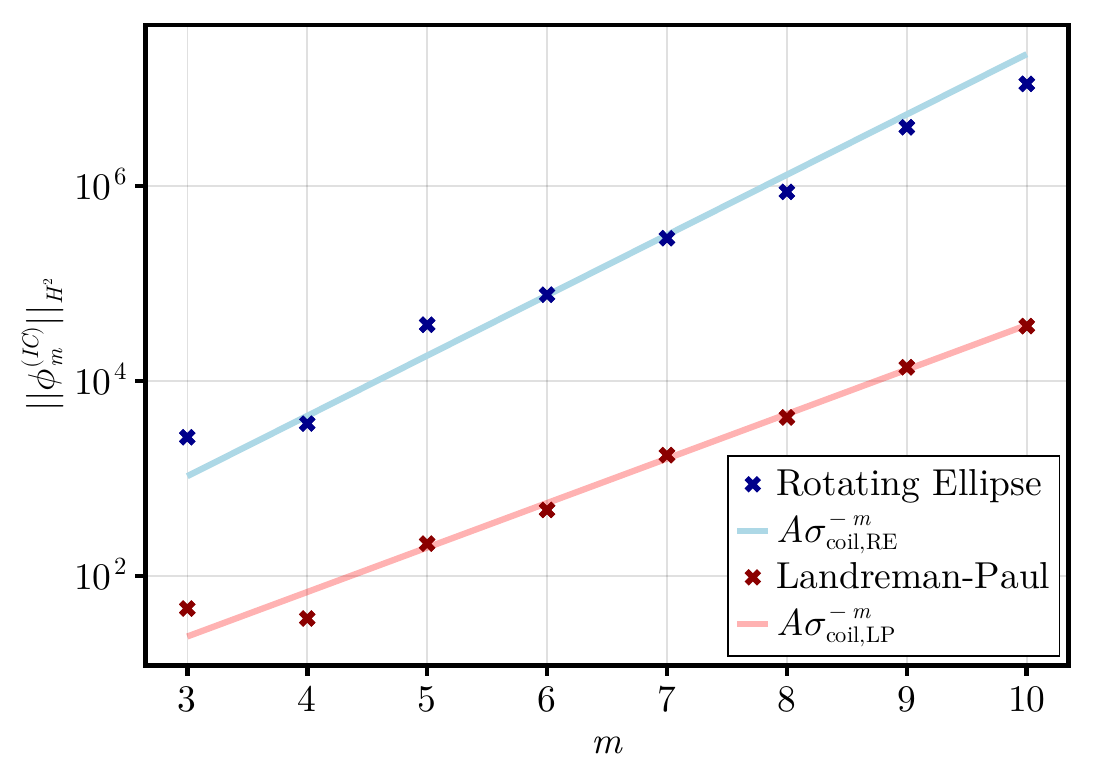}
    \caption{Plot of (markers) the coefficient norm $\lVert\phiic_m\rVert_{H^2}$ versus the order $m$ and (lines) best-fit lines $A \sigma_{\mathrm{coil}}^{-m}$ where $\sigma_{\mathrm{coil}}$ is the axis-to-coil distance.
    }
    \label{fig:phiic}
\end{figure}
To begin our analysis of the examples, we consider the inputs to the near-axis expansion.
Corollary \ref{cor:convergence} suggests that there are two length scales dictated by the inputs.
First is the radius of curvature of the axis.
Letting $\Sigma = \min_s \kappa^{-1}(s)$ be the radius of curvature, we find
\begin{equation*}
    \Sigma_{RE} = 0.987, \qquad \Sigma_{LP} = 0.681.
\end{equation*}
The other length scale of interest is the radius of convergence $\sigma_0$ of $\phiic$ in the $(\sigma_0,H^{q+2D})$-analytic norm.
However, to the orders we compute to, this radius appears to depend on the exponent $q+2D$.
To determine the most informative exponent, we consider the work by \citet{kappel2024}, where it was shown that the normalized gradient of the magnetic field is a strong predictor for plasma-coil separation.
Because the $H^2(\Tbb^2)$ norm measures the size of the second derivative of $\phiic_m$ (and therefore the gradient of the input magnetic field), we conjecture this corresponds to the most practical exponent.

To verify this, we first compute the minimum axis-to-coil distance $\sigma_{\mathrm{coil}}$ for both configurations to be
\begin{equation*}
    \sigma_{\mathrm{coil,RE}}=0.241, \qquad \sigma_{\mathrm{coil,LP}} = 0.350.
\end{equation*}
We note that $\sigma_{\mathrm{coil}}<\Sigma$ for both configurations, indicating that the distance-to-coil is the limiting factor for convergence (cf.~Cor.~\ref{cor:convergence}).
In Fig.~\ref{fig:phiic}, we plot the $H^2$ norms of $\phiic_m$ vs $A \sigma_{\mathrm{coil}}^{-m}$, where the coefficient $A$ is found via a best fit for both configurations.
In both cases, we find there is remarkable agreement, indicating that the $H^2$ radius of convergence of $\phiic$ could be used as a proxy for distance-to-coils.

\subsection{Magnetic Field Convergence}

\begin{figure}
    \centering
    \includegraphics[width=1.0\linewidth]{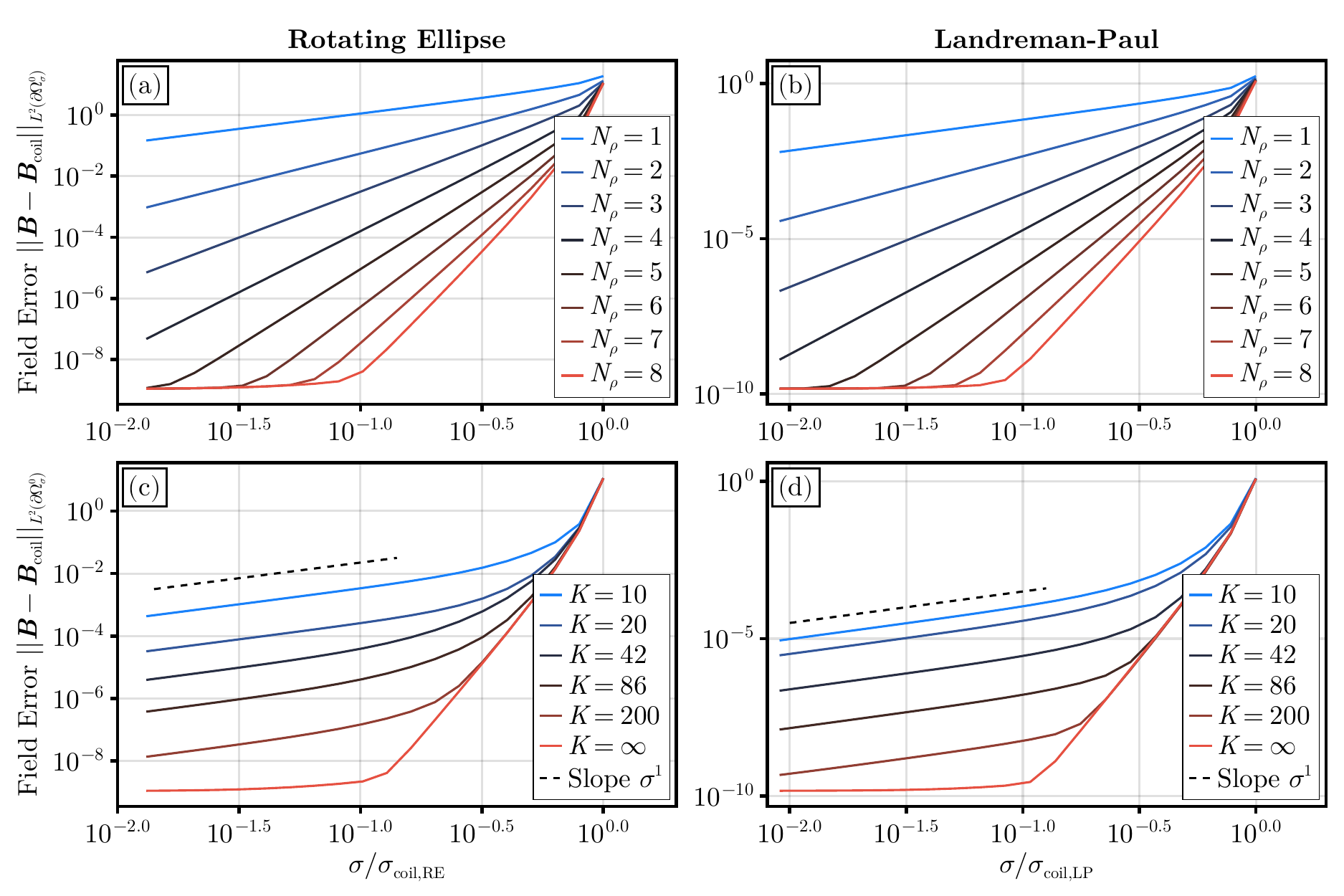}
    \caption{(a-b) The error \eqref{eq:B-error} as a function of the normalized distance from axis $\sigma/\sigma_{\mathrm{coil}}$ for varying orders of approximation $\Nrho$.
    (c-d) The error \eqref{eq:B-error} as a function of $\sigma/\sigma_{\mathrm{coil}}$ for varying values of the regularization parameter $K$ ($K=\infty$ is unregularized). 
    }
    \label{fig:Bfield}
\end{figure}

Next, we compute the near-axis expansion via the procedure in Box \ref{eq:reg-box}.
We perform the expansion both without regularization ($P=0$) and with the regularization operator in \eqref{eq:P-reg}.
For the regularized runs, we use $D=2$ throughout and vary $K$ between $10$ and $200$ to assess how the equilibrium changes between strong and weak regularization respectively.

As a first test of the output convergence, we compare the coil magnetic field against the unregularized expansion.
Because the input is harmonic, we expect that the near-axis expansion will converge from Prop.~\ref{prop:harmonic} (unless floating-point errors overwhelm the solution, which is not observed to this order). 
We verify this by computing the $L^2$ magnetic field error on surfaces about the axis
\begin{equation}
\label{eq:B-error}
    \norm{\Bbf - \Bbf_{\mathrm{coil}}}_{L^2(\partial \Omega_\sigma^0)} = \frac{1}{4\pi^2} \int_0^{2\pi} \int_{0}^{2\pi} \abs{\Bbf(\sigma,\theta,s) - \Bbf_{\mathrm{coil}}(\sigma,\theta,s)}^2 \dif \theta \df s,
\end{equation}
where $\Bbf$ is computed from the near-axis expansion and $\Bbf_{\mathrm{coil}}$ is computed directly from Biot-Savart.

In Fig.~\ref{fig:Bfield} (a-b), we plot the error \eqref{eq:B-error} versus the normalized distance-from-axis $\sigma/\sigma_{\mathrm{coil}}$.
This error is computed for the approximation
\begin{equation*}
    \phi_{\leq \Nrho} \approx \sum_{m=2}^{\Nrho} \phi_m \rho^m, \qquad \Bbf_{\leq \Nrho-1} = \nabla \phi_{\leq \Nrho} + \tilde{\Bbf}_{\leq \Nrho-1}
\end{equation*}
where $\Nrho$ is varied from $1$ to $8$. 
For both configurations, we find that the error of the magnetic field obeys the expected power law
\begin{equation*}
    \norm{\Bbf_{\leq \Nrho-1} - \Bbf_{\mathrm{coil}}}_{L^2(\partial \Omega_\sigma^0)} = \Ocal(\sigma^{-\Nrho}).
\end{equation*}
The error curves for varying $\Nrho$ meet at $\sigma = \sigma_{\mathrm{coil}}$, indicating that the output radius of convergence is limited by the coils.
This tells us that the limit of convergence $\sigma = \sigma_0$ is achievable in Corollary \ref{cor:convergence}.

Turning to the effects of regularization, we fix $\Nrho=9$ and plot the error \eqref{eq:B-error} versus $\sigma/\sigma_{\mathrm{coil}}$ for varying $K$ between $10$ and $200$ in Fig.~\ref{fig:Bfield} (c-d).
We also include the unregularized solution, labeled with $K=\infty$. 
We find that as the regularization becomes stronger ($K$ decreases), the magnetic field loses fidelity near the core.
We attribute to the increasing loss of accuracy of the high-wavenumber $s$ modes, while the low-wavenumber modes maintain accuracy.
Then, far from the axis, the regularized error inflects to begin to agree with the rate of convergence of the unregularized solution.
So, while the solution loses a high-wavenumber fidelity, the low wavenumbers maintain a similar level of accuracy.
Comparing Figs.~\ref{fig:Bfield} (a-b) to (c-d), we see that a regularized high-order expansion can achieve an equivalent error to an unregularized lower-order expansion near the axis while maintaining that fidelity far from the axis.

To address the role of regularization more fully, however, we need to consider how the fidelity of the expansion on less tuned inputs. 
To do this, we perturb $\rbf_0$, $B_0$, and $\phiic$ by the random functions as
\begin{align}
\nonumber
    \delta r_{0,i} &= \epsilon \sum_{\ell = -\Ns}^{\Ns} \frac{X_{i\ell}}{1 + (\ell/K_{\epsilon})^{4+q}} , \\
\nonumber
    \delta B_0 &=  \epsilon \sum_{\ell = -\Ns}^{\Ns}\frac{Y_{\ell}}{1 + (\ell/K_{\epsilon})^{1+q}} , \\
\label{eq:perturbations}
    \delta \phiic_{mn} &= \epsilon \sum_{\ell = -\Ns}^{\Ns} \frac{Z_{mn\ell}}{1 + (m/K_{\epsilon})^{q+2D} + (\ell/K_{\epsilon})^{q+2D}} , && \quad n\in \{0,m\}.
\end{align}
where $X_{i\ell}$, $Y_\ell$, and $Z_{mn\ell}$ are i.i.d.~unit normal random variables, $q = 2$, and $\epsilon = 10^{-6}$ for the rotating ellipse and $\epsilon=10^{-4}$ for Landreman-Paul.
We have chosen the regularity of the perturbation to align with the inputs in Box \eqref{eq:reg-box}.

\begin{figure}
    \centering
    \includegraphics[width=1.0\linewidth]{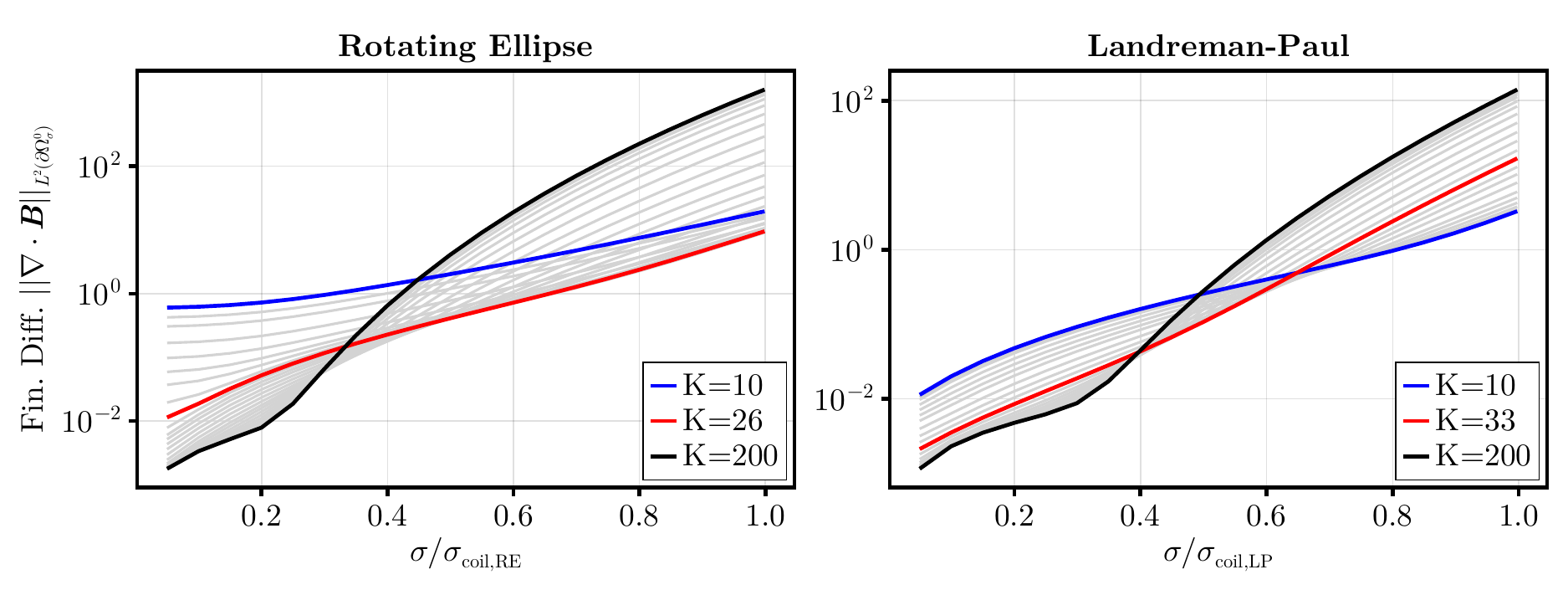}
    \caption{Finite difference residual \eqref{eq:FD-resid} as a function of the normalized distance from axis $\sigma/\sigma_{\mathrm{coil}}$ for the perturbed rotating ellipse and Landreman-Paul inputs (see Eq.~\ref{eq:perturbations}).
    For both plots, three lines are colored and labeled, while the gray lines represent other values of $K$ interpolating between $K=10$ and $K=200$. 
    }
    \label{fig:finite-diff}
\end{figure}

In Fig.~\ref{fig:finite-diff}, we consider the accuracy of the solution to the perturbed problem for $K$ varying between $10$ and $200$ for both examples with $N_\rho = 9$ fixed. 
To measure the accuracy, we no longer have a coil set to compare the solution directly. 
So, we instead measure the residual of Poisson's equation
\begin{equation}
\label{eq:FD-resid}
    \norm{\nabla \cdot \Bbf}_{L^2(\partial \Omega_\sigma^0)} = \frac{1}{4\pi^2} \int_0^{2\pi} \int_{0}^{2\pi} \abs{\nabla \cdot (\nabla \phi + \tilde{\Bbf})}^2 \dif \theta \df s,
\end{equation}
where we evaluate every derivative (including in the metric) via finite differences. 
For both perturbed examples, the best solution near the axis is the lightly regularized $K=200$ solution.
However, beyond a certain radius between $0.3\sigma_{\mathrm{coil}}$ and $0.4 \sigma_{\mathrm{coil}}$, more regularized solutions improve upon the less regularized ones in the finite difference metric. 
For our examples, we find that $K=26$ for the rotating ellipse and $K=33$ for Landreman-Paul are perhaps the best choices in practice.
This figure potentially indicates a more general principle: the further from the axis one wants accuracy of the expansion, the more regularized the expansion likely has to be.

\subsection{Magnetic Coordinate Convergence}
\begin{figure}
    \centering
    \includegraphics[width=1.0\linewidth]{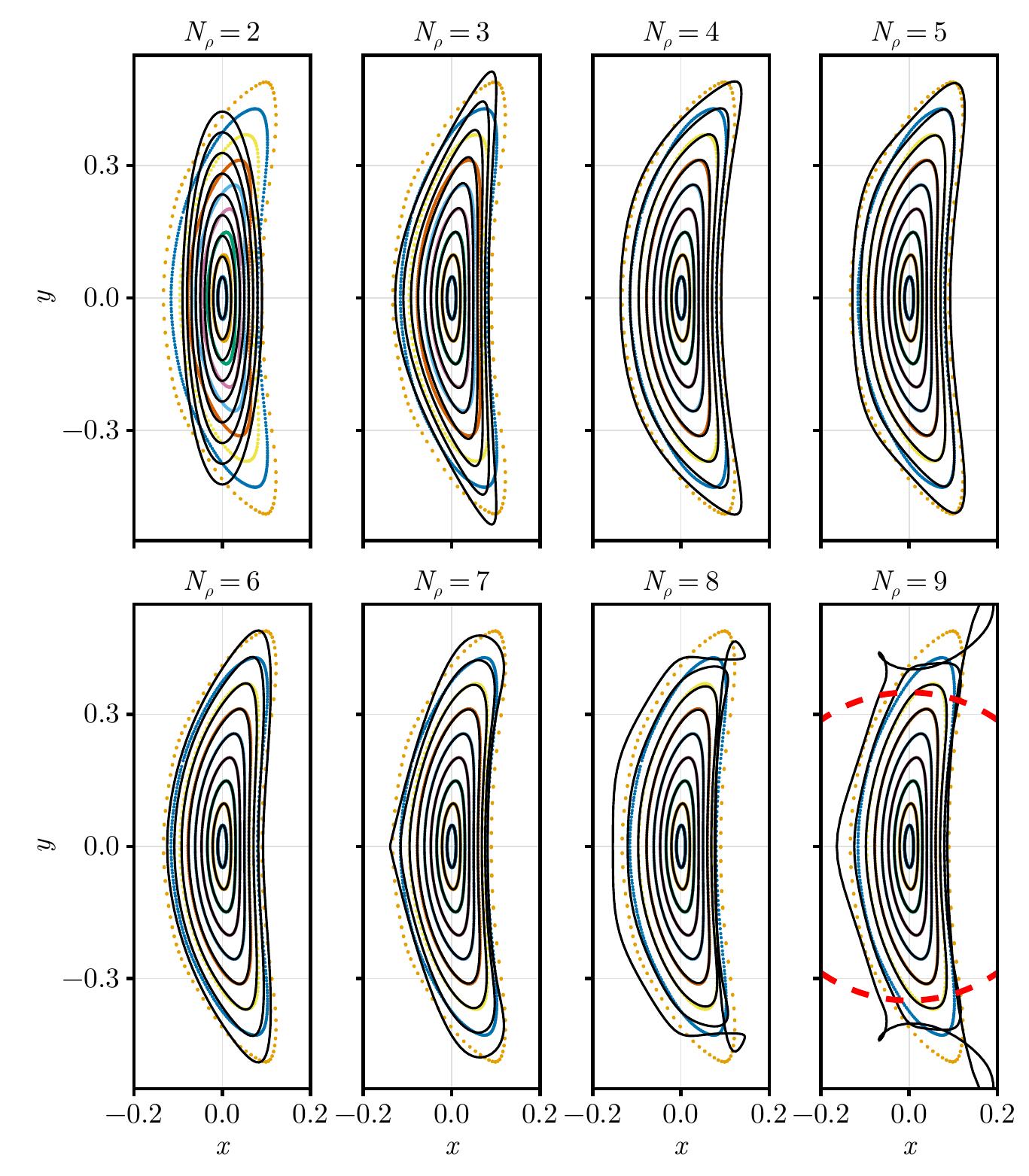}
    \caption{(black) Near-axis approximations of flux surfaces for varying orders of approximation $N_\rho$; (color) a Poincar\'e plot of the true coil magnetic field lines. 
    In the final $N_\rho = 9$ panel, we plot a circle with radius $\sigma_{\mathrm{coil,LP}}$ in red.
    }
    \label{fig:surfaces}
\end{figure}

To compute straight field-line coordinates, we return to the unregularized Landreman-Paul configuration.
Then, using the straight field-line magnetic coordinate equation from Box \ref{eq:fieldline-box}, we compute the approximate coordinates $(\xi,\eta)$ for $N_\rho$ varying between $2$ and $9$, where we note the $N_\rho=2$ approximation of $\phi$ provides the leading-order field-line behavior.
To find flux surfaces, we then invert $(\xi,\eta)$ to find the distance-to-axis coordinates $(x(\xi,\eta,s),y(\xi,\eta,s))$, where magnetic surfaces are parameterized by $\xi^2 + \eta^2 = \psi$.

In Fig.~\ref{fig:surfaces}, we plot in black the computed surfaces on the $s=0$ Poincar\'e section for varying values of $N_\rho$. 
For comparison, we plot the intersections of coil magnetic field lines in the background.
At leading order, we see the surfaces are elliptical, while higher orders account for more shaping in the $\theta$ direction. 
Then, as the order increases beyond $5$, the surfaces surfaces away from the core start diverging.

To investigate this divergence, we plot a red circle of constant radius $\sigma_{\mathrm{coil,LP}}$ in the $N_\rho = 9$ panel.
We see that the circle appears to separate the divergent surfaces from the convergent ones.
We believe this is the likely reason for the divergence, however there are still other possibilities.

\begin{figure}
    \centering
    \includegraphics[width=1.0\linewidth]{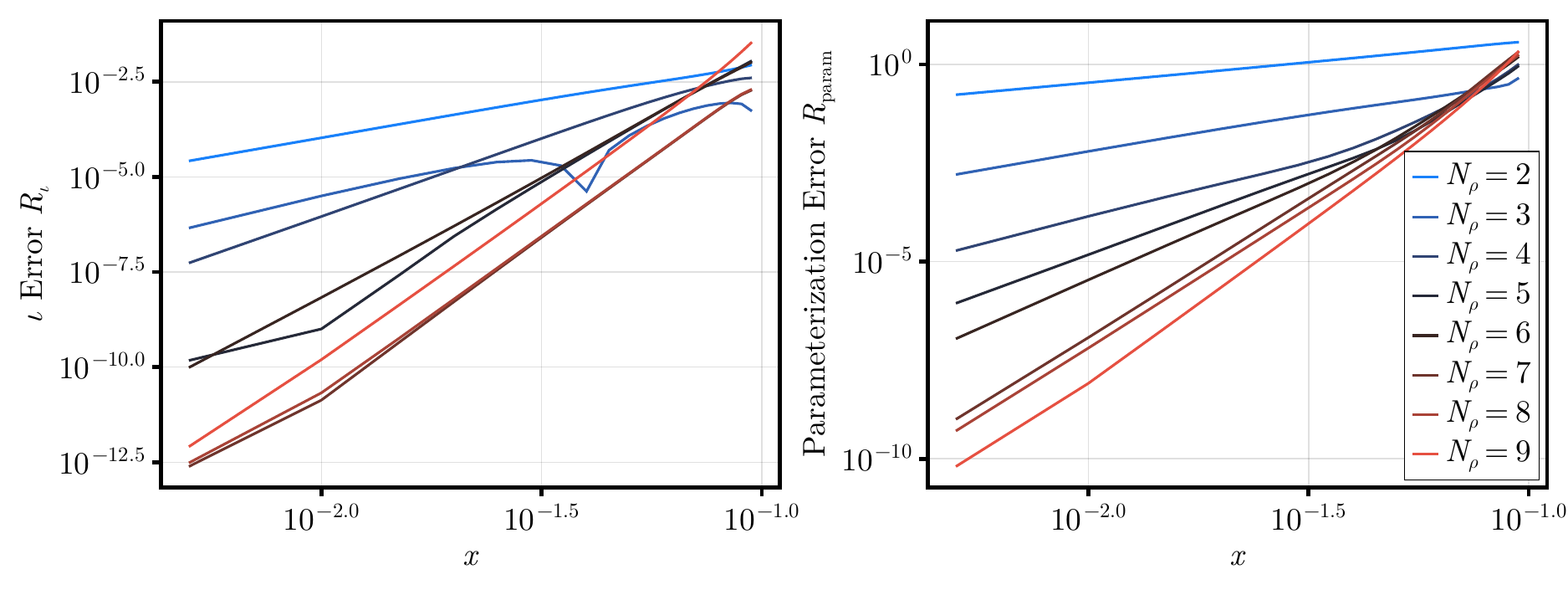}
    \caption{(left) The rotational transform $R_\iota$ and (right) the parameterization error $R_{\mathrm{param}}$ defined in Eq.~\ref{eq:Riota} as a function of the inboard $x$ distance (cf.~Fig.~\ref{fig:surfaces}) for varying orders of approximation $N_\rho$. 
    }
    \label{fig:surface-conv}
\end{figure}

To assess the errors of surfaces closer to axis, we turn to a more quantitative measure.
To do this, we first use the method from \citet{ruth_finding_2024} in the \texttt{SymplecticMapTools.jl} package \citep{smt-jl} to compute invariant circles $(x_{\mathrm{coil}}(\theta), y_{\mathrm{coil}}(\theta))$ and the rotational transform $\iota_{\mathrm{coil}}$ on the cross section from the Poincar\'e plot trajectories.
Then, as a function of the inboard distance $x$ from the axis (see Fig.~\ref{fig:surfaces}), we compute rotational transform and parameterization errors as
\begin{align}
\label{eq:Riota}
    R_{\iota} &= \abs{\iota - \iota_{\mathrm{coil}}}, \\ 
\nonumber
    R_{\mathrm{param}} &= \left[\int_{0}^{2\pi} (x(\psi(x,0,0),\theta,0)-x_{\mathrm{coil}}(\theta))^2 + (y(\psi(x,0,0),\theta,0)-y_{\mathrm{coil}}(\theta))^2 \dif \theta \right]^{1/2}.
\end{align}
In Fig.~\ref{fig:surface-conv}, we plot both errors with varying $N_\rho$.
In both cases, the rotational transform and parameterization converge to high accuracy near the core.
However, they begin to diverge before the the outermost surface, agreeing with the visual divergence in Fig.~\ref{fig:surfaces}. 

\section{Conclusion}
\label{sec:conclusion}
In this paper, we have investigated the convergence of the near-axis expansion in vacuum, both theoretically and numerically.
From the theoretical point of view, we showed in Theorem \ref{thm:illposed} that the near-axis expansion is ill-posed, even in the relatively simple case of vacuum fields.
However, as shown in Theorem \ref{thm:reg-theorem}, we found the near-axis problem can be regularized giving a guarantee of convergence for appropriately smooth input data.
In particular, this tells us that a truncated near-axis expansion is an approximation to the solution of the regularized problem.
Combining this with Proposition \ref{prop:a-posteriori}, we find can estimate the error of the regularized expansion from a true solution.

From the numerical results, we have verified that the near-axis expansion can converge in vacuum.
This includes convergence of surfaces, where we have shown that the rotational transform and surface parameterizations can be approximated near the axis to high accuracy. 
Moreover, we demonstrated that the radius of convergence of the expansion is directly tied to the minimum distance to coils.
Under perturbation, we found that the regularization reduces the residual of Poisson's equation far from the axis.

Our analysis suggests that the following four quantities should be kept in mind for future optimization problems:
\begin{itemize}
    \item The axis, on-axis field, and higher moments should all be sufficiently regular for the expansion to converge (see Box~\ref{eq:reg-box}). These can be enforced, e.g., by Sobolev norms on the inputs of the near-axis expansion.
    \item In particular, the $H^2$-norm radius of convergence of $\phiic$ appears to indicate the distance to coils from the axis (see Fig.~\ref{fig:phiic}). This gives a potential metric for plasma-coil distance.
    \item The axis curvature also limits to the radius of convergence, so this should be small relative to the desired minimum distance to coils.
    \item In the case that the above terms are not sufficient, the error in Proposition \ref{prop:a-posteriori} could be used to monitor the accuracy of the solutions.
\end{itemize}
Using these metrics, a moderate-order near-axis expansion (say, $4\leq N \leq 6$) could be used to explore the space of stellarators more effectively.
This could allow for the use of new near-axis optimization problems.

Looking forward, these results indicate that regularization is likely also required for the near-axis expansion to converge in pressure.
The form of equation \ref{eq:fictitious-current} gives a potential path forward, where the regularization could be expressed as a fictitious current.
Physically, a link between regularization and extended MHD models that provide additional current contributions can be studied.
However, the issue of small denominators for near-rational $\iota_0$ (see Eq.~\eqref{eq:xi-coefficient-formula}) will appear, which will combine with the regularized expansion in a non-trivial way in pressure.
It remains to be seen whether regularization can be used for improved convergence of these surfaces.

\acknowledgments{The authors would like to acknowledge the helpful conversations with Tony Xu, Sean Yang, Gokul Nair, and Joshua Burby in the development of this work.}

\paragraph{\textbf{Funding.}} 
This work was supported by a grant from the Simons Foundation (No. 560651, D B) and by DOE DE-SC0024548.
This material is based upon work supported by the National Science Foundation under Grant No.2409066. Any opinions, findings, and conclusions or recommendations expressed in this material are those of the author(s) and do not necessarily reflect the views of the National Science Foundation.

\paragraph{\textbf{Declaration of Interests.}} The authors report no conflict of intererst.

\paragraph{\textbf{Data availability statement.}} The data that support the finding of this study are openly available in the package StellaratorNearAxis.jl at the GitHub repository \verb|https://github.com/maxeruth/StellaratorNearAxis.jl|.

\appendix

\section{Proofs}
\label{app:proofs}
\subsection{Proof of Proposition \ref{prop:harmonic}}
\label{subsec:prop-harmonic-proof}
Because $\Bbf$ is a vacuum field, its Cartesian components $(B^1,B^2,B^3)$ are also harmonic and therefore real-analytic, meaning at each point on the axis $\rbf_0(s)$ it has a uniformly convergent Taylor series in a ball of size $(\sqrt{2}-1)\sigma$ \citep[][Theorem 1.28]{axler_harmonic_2001}. 
By choosing the coefficients $\phi_{mn}$ to match the Taylor series at each point, we find that the near-axis expansion is uniformly convergent near the axis.
Because $\Bbf$ is harmonic the coefficients must satisfy the near-axis problem \eqref{eq:nae-box}.
Finally, because the solution to the near-axis expansion is unique (Prop.~\ref{prop:formal-solvability}), the proposition is proven.

\subsection{Proof of Proposition \ref{prop:continuity}}
\label{subsec:prop-continuity-proof}
We start with a lemma on derivatives on $(\sigma,W)$-analytic functions.
\begin{lemma}
\label{lemma:xy-deriv}
    Let $(x,y) = (\rho \cos \theta, \rho \sin \theta)$ and $W$ be $\Ccal^q(\Tbb^2)$ or $H^q(\Tbb^2)$ for $q\geq0$. 
    The derivatives $\pd{}{x}$ and $\pd{}{y}$ are bounded operators from $(\sigma,W)$-analytic functions to $(\sigma',W)$-analytic functions for all $0<\sigma'<\sigma$.
\end{lemma}
\begin{proof}
    We will prove this for the $x$ derivative, as the proof for the $y$ derivative is identical.
    First, we observe that $x$ derivatives preserve the analytic structure \eqref{eq:analytic-form}.
    Let $f$ be $(\sigma,W)$-analytic where $W$ is $H^q$ or $\Ccal^q$ for $q \geq 1$. 
    In polar coordinates, we have 
    \begin{equation*}
        \pd{f}{x} = \cos \theta \pd{f}{\rho}- \frac{\sin \theta}{\rho} \pd{f}{\theta}.
    \end{equation*}
    For both $\rho$ and $\theta$ we have 
    \begin{equation*}
        \norm{\pd{}{\rho}(\rho^m f_m)}_{W} \leq m \rho^{m-1} \norm{f_m}_{W}, \qquad \norm{\frac{1}{\rho}\pd{}{\theta}(\rho^m f_m)}_{W} \leq m \rho^{m-1} \norm{f_m}_{W}.
    \end{equation*}
    Multiplying by $\sin \theta$ and $\cos \theta$ is bounded on both $H^q$ and $\Ccal^q$, so
    \begin{align*}
        \norm{\left(\pd{f}{x}\right)_{m}}_W &\leq C (m+1) \norm{f_{m+1}}_W , \\
        &\leq C (m+1) \norm{f}_{\sigma,W} \sigma^{-m-1}, \\
        &\leq \frac{C(m+1)(\sigma')^m}{\sigma^{m+1}} \norm{f}_{\sigma,W} (\sigma')^{-m} 
        \leq C \norm{f}_{\sigma,W} (\sigma')^{-m}.
    \end{align*}
\end{proof}

\begin{lemma}
\label{lemma:s-derivatives}
    Let $q\geq 1$. The derivatives $\pd{}{s}$ and $\pd{}{\theta}$ are bounded operators from 
    \begin{enumerate}
        \item $(\sigma,\Ccal^q)$-analytic functions to $(\sigma,\Ccal^{q-1})$-analytic functions and,
        \item $(\sigma,H^q)$-analytic functions to $(\sigma,H^{q-1})$-analytic functions.
    \end{enumerate}
\end{lemma}
\begin{proof}
    Simply notice $\pd{}{s}$ and $\pd{}{\theta}$ are bounded from $\Ccal^q$ to $\Ccal^{q-1}$ and from $H^{q}$ to $H^{q-1}$.
\end{proof}

Combining the two above lemmas, if we choose $\sigma'<\sigma$, $q \geq 0$, and $f$ to be a $(\sigma,\Ccal^q)$-analytic function, then for all $m,n,\ell>0$ such that $m+n+\ell\leq q$, the function
\begin{equation*}
    g = \pd{f}{x^n \partial y^m \partial s^\ell}
\end{equation*}
is $(\sigma', \Ccal^0)$-analytic.
So, it suffices to prove that $g$ is continuous in $\Omega_{\sigma'}^0$.

Let $0<\tilde{\rho} < \sigma'$ and $(\tilde{\theta},\tilde{s})\in\Tbb^2$. 
Then, choose $\rho^*$ such that $\tilde{\rho} < \rho^*<\sigma'$. 
We will show that $g$ is continuous at the point $(\tilde{\rho},\tilde{\theta},\tilde{s})$.
Letting let $\rho_1,\rho_2\leq \rho^*$, we can establish a Lipschitz bound of $g$ in $\rho$:
\begin{align*}
    \abs{g(\rho_1, \theta, s) - g(\rho_2,\theta,s)} &\leq \sum_{m=\beta}^\infty \norm{g_m}_{C} \abs{\rho_1^{m} - \rho_2^{m}}, \\
    &\leq G\frac{\abs{\rho_1-\rho_2}}{\rho^*} \sum_{m=1}^{\infty}m \left(\frac{\rho^*}{\sigma}\right)^{m} \leq L \abs{\rho_1-\rho_2}.
\end{align*}
Now, let $(\rho_m,\theta_m,s_m) \to (\tilde{\rho},\tilde{\theta},\tilde{s})\in \Omega_\sigma^0$ and $\sup(\rho_m) < \rho^* < \sigma$. We have
\begin{equation*}
    \abs{ g(\rho_m,\theta_m,s_m) - g(\tilde{\rho},\tilde{\theta},\tilde{s})} 
     \leq L \abs{\rho_m-\rho} + \abs{ g(\tilde{\rho},\theta_m,s_m)-  g(\tilde{\rho},\tilde{\theta},\tilde{s})}.
\end{equation*}
Because surfaces of $g$ converge in $C$, both terms converges to zero giving our result.

\subsection{Proof of Theorem \ref{thm:illposed}}
\label{subsec:illposed-proof}
To prove this theorem, we will begin with a few facts about operators on $(\sigma,q)$-analytic functions. 
\begin{lemma}
\label{lemma:analytic-multiplication}
    Let $f$ be $(\sigma,H^q)$-analytic and $g$ be $(\Sigma,\Ccal^q)$-analytic where $0 < \sigma \leq \sigma_0 < \Sigma$ and $q\geq 0$. Then $fg$ is $(\sigma,H^q)$-analytic with
    \begin{equation*}
        \norm{fg}_{\sigma,q} \leq C \norm{g}_{\Sigma,\Ccal^q} \norm{f}_{\sigma, H^q}, 
    \end{equation*}
    where $C$ only depends on $q$, $\sigma_0$, and $\Sigma$. 
\end{lemma}
\begin{proof}
    The coefficients of $fg$ are
    \begin{equation*}
        (fg)_m = \sum_{n=0}^m f_n g_{m-n}.
    \end{equation*}
    We can bound the norm of $f_n g_{m-n}$ as
    \begin{equation*}
        \norm{f_n g_{m-n}}_{H^q} \leq C \norm{f_n}_{H^q} \norm{g_{m-n}}_{\Ccal^q} \leq C \norm{f}_{\sigma,H^q} \norm{g}_{\Sigma,\Ccal^q} \sigma^{-n} \Sigma^{-(m-n)},
    \end{equation*}
    where the constant $C$ only depends on $q$. 
    So,
    \begin{align*}
        \norm{(fg)_m}_{H^q} &\leq C \norm{f}_{\sigma,H^q} \norm{g}_{\Sigma,\Ccal^q} \sigma^{-m} \sum_{n = 0}^m \left(\frac{\sigma}{\Sigma}\right)^{-n},\\
        &\leq C \norm{f}_{\sigma,H^q} \norm{g}_{\Sigma,\Ccal^q} \sigma^{-m} \sum_{n = 0}^\infty \left(\frac{\sigma_0}{\Sigma}\right)^{-n},\\
        &\leq C \norm{f}_{\sigma,H^q} \norm{g}_{\Sigma,\Ccal^q} \sigma^{-m} \frac{1}{1-\sigma_0/\Sigma},
    \end{align*}
    giving the result.
\end{proof}

\begin{lemma}
\label{lemma:general-operator-bounded}
    Let $a^\alpha$ be $(\Sigma,\Ccal^{q})$-analytic for all degree 3 multi-indices $\abs{\alpha}\leq m$, $0<\sigma'<\sigma<\Sigma$, and $q\geq 0$.
    The operator 
    \begin{equation*}
        L = \sum_{\abs{\alpha}\leq m} a^\alpha \pd{^{\abs{\alpha}}f}{x^{\alpha_1} \partial y^{\alpha_2} \partial s^{\alpha_3}}
    \end{equation*}
    is bounded from $(\sigma,H^{q+m})$-analytic functions to $(\sigma',H^q)$-analytic functions.
\end{lemma}
\begin{proof}
    Combine Lemmas \ref{lemma:xy-deriv}, \ref{lemma:s-derivatives}, and \ref{lemma:analytic-multiplication}.
\end{proof}

\begin{corollary}
    Let $\bm r_0$ be $\Ccal^{4+q}$ for $q \geq 0$ and $\ell',\kappa>0$.
    There exists a $\Sigma > 0 $ such that the Laplacian, defined in $(x,y,s)$ coordinates via the left-hand-side of \eqref{eq:laplace-coordinate}, is a bounded operator from $(\sigma,H^{q+2})$ to $(\sigma',H^q)$. 
    Similarly, the divergence operator, defined in $(x,y,s)$ coordinates via the right-hand-side of \eqref{eq:laplace-coordinate}, is a bounded operator from $(\sigma,H^{q+1})$ to $(\sigma',H^q)$. 
\end{corollary} 
\begin{proof}
\label{cor:Delta-boundedness}
    In $(x,y,s)$-coordinates, we have $h_s=1-\kappa x$, $\sqrt{g}=\ell'h_s$, and
    \begin{equation*}
        g^{-1} = \frac{1}{h_s^2}\begin{pmatrix}
            h_s^2 + \tau^2y^2  & - \tau^2x y & \frac{\tau y}{\ell'} \\
            -\tau^2 xy & h_s^2 + \tau^2 x^2 & \frac{-\tau x}{\ell'} \\
            \frac{\tau y}{\ell'} & -\frac{\tau x}{\ell'} & \frac{1}{(\ell')^2}
        \end{pmatrix}.
    \end{equation*}
    Because $\bm r_0 \in \Ccal^{4+q}$, $\ell' \in \Ccal^{3+q}$, $\kappa \in \Ccal^{2+q}$ and $\tau \in \Ccal^{1+q}$, so $g^{-1} \in \Ccal^{1+q}$.
    This means the elements of $h_s^2 g^{-1}$ are $(\sigma,\Ccal^{1+q})$-analytic with finite series.
    To include the factor of $h_s^{-2}$, we have
    \begin{equation*}
        \frac{1}{h_s} = \sum_{m=0}^{\infty}(\kappa \cos \theta)^m \rho^m.
    \end{equation*}
    The function $\kappa \cos \theta$ is in $\Ccal^{2+q}$, so $\norm{(\kappa \cos \theta)^m}_{H^{2+q}} \leq C^m \norm{\kappa \cos \theta}_{\Ccal^{2+q}}^m$ for some $C\geq 1$ and $h_s^{-1}$ is $(\Sigma,\Ccal^{2+q})$ for some $0<\Sigma \leq \norm{\kappa}^{-1}_{C^0}$. 
    Finally, after performing the chain rule to bring the operator to the form in Lemma \ref{lemma:general-operator-bounded}, the coefficients are each in $(\Sigma, H^{q})$-analytic, giving the result.
    The same argument applies to the divergence operator.
\end{proof}

The last ingredient needed for the proof of ill-posedness is Cauchy's Estimates for harmonic functions:
\begin{theorem}[Cauchy's Estimates \citep{axler_harmonic_2001}]
\label{thm:CauchysEstimates}
    Let $\alpha = (\alpha_1,\alpha_2, \dots, \alpha_d)$ be a multi-index. Then for some constant $C_\alpha>0$, all harmonic functions $\phi$ bounded by $M$ on the radius-$R$ ball $B(x,R)$ satisfy the inequality
    \begin{equation*}
        \abs{D^\alpha \phi(x)} \leq \frac{C_\alpha M}{R^{\abs{\alpha}}}.
    \end{equation*}
\end{theorem}

Now, we return to the proof of theorem \ref{thm:illposed}.
It is sufficient to show that $\phi$ is not $(\sigma,H^q)$-analytic for $q = 2$. 
So, consider input data such that $\norm{\phi}_{\sigma,H^2} < \infty$, say constructed via proposition \ref{prop:harmonic}. 
We will focus on perturbations in $B_0$ of the form $\Delta B_0^{(j)} = c_j e^{\I j s}$ where $c_j = j^{-N}$ and $N>q_0+1$ is a positive integer, while $\bm r_0$ and $\phiic$ remain constant. 
Clearly, $\norm{\Delta B_0}_{H^{1+q_0}} \to 0$ as $j\to \infty$.
Let $\phi^{(j)}$ be the formal power series solution of the perturbed problem.

We want to show that $\phi^{(j)}$ cannot converge to a $(\sigma,H^q)$-analytic solution.
In case that the perturbed near-axis expansion for fixed $j$ does not converge in $\Omega_\sigma$, then $\norm{\phi^{(j)}}_{\sigma,H^2} = \infty$ for all $\sigma$ and we are done, as the operator fails to be bounded for a specific function.
Otherwise, suppose the near-axis expansion solution converges on $\Omega_\sigma$ and each $\phi^{(j)}$ is $(\sigma,H^2)$-analytic.
Then using Corollary \ref{cor:Delta-boundedness}, for $\sigma'<\sigma$, $\Delta \phi^{(j)} + \nabla \cdot \tilde{\Bbf}$ is a $(\sigma',L^2)$-analytic function. 
Because $\phi$ satisfies the near-axis expansion, the solution must satisfy $(\Delta\phi^{(j)}+\nabla \cdot \tilde{\Bbf})_m =0$ in $L^2(\Tbb^2)$, further implying that $\Delta\phi^{(j)}+\nabla \cdot \tilde{\Bbf} = 0$ in $L^2(\Omega_{\sigma'}^0)$.
Pulling this back to $\Omega_{\sigma'}$, we are solving the standard Poisson's equation $\Delta \phi^{(j)} = -\nabla \cdot \tilde{\Bbf}$ in $L^2(\Omega_{\sigma'})$. 
If we locally define 
\begin{equation*}
    \psi^{(j)} = \phi^{(j)} + \int_0^s B_0(s')\ell'(s') \dif s', 
\end{equation*}
we find $\Delta \psi^{(j)} = 0$.
That is, $\psi^{(j)}$ is analytic in simply connected subdomains of $\Omega_{\sigma'}$ and the magnetic field is locally the gradient of $\psi^{(j)}$. 

Then, let $B(x,R)$ be a ball around a point on the magnetic axis $\bm r_0(s_0)$.
At $\bm r_0(s_0)$, the order $M>N$ derivative in the tangent direction of $\bm r_0$ of $\psi^{(j)}$ takes the polynomial form
\begin{equation*}
    \pd{^M\psi^{(j)}}{s^M} = \pd{^M B_0}{s^M} + c_j \sum_{\ell = 0}^M a_\ell j^\ell,
\end{equation*} 
where $a_{M} = (\I / \ell'(s_0))^{M} e^{\I j s_0} \neq 0$ and $a_\ell$ for $\ell<M$ contain higher order derivatives of the axis.
As such, there is a $J$ such that $j > J$ implies that
\begin{equation*}
    \abs{\pd{^M\psi^{(j)}}{s^M}} > \frac{\abs{a_N}}{2} j^{M-N}.
\end{equation*}
Then, by theorem \ref{thm:CauchysEstimates}, we have that 
\begin{equation*}
    \max_{B(x,R)} \psi^{(j)} > \frac{R^N a_N}{2 C_N} j^{M-N}.
\end{equation*}
So, as $j \to \infty$, $\psi^{(j)}$ cannot converge to a continuous function.
However, because $\phi$ was assumed $(\sigma,H^2)$-analytic and by Corollary \ref{cor:Hq-continuity} it must be continuous, we have drawn a contradiction.

\subsection{Proof of Theorem \ref{thm:reg-theorem}}
\label{sec:reg-proof}
We will start with two lemmas. The first is on the boundedness of the right-hand-side operator of the PDE \eqref{eq:arbitrary-PDE}:
\begin{lemma}
\label{lemma:La}
    Let $q\geq 0$, $D\geq 1$, $0<\sigma\leq \sigma_0 < \Sigma$, $\phi$ be $(\sigma,H^{q+2D})$-analytic, and $a^{(j)}$ be $(\Sigma,\Ccal^q)$-analytic for $1\leq j \leq 7$. 
    Then, the operator $L^{(a)}$ defined in \eqref{eq:La} preserves the analytic form \eqref{eq:analytic-form} and satisfies the bound
    \begin{equation*}
        \norm{(L^{(a)}\phi)_m}_{H^q} \leq C (m+1) \left(\sum_{j=1}^7 \norm{\ac{j}}_{\Sigma,\Ccal^q} \right) \norm{\phi}_{\sigma,H^{q+2D}} \sigma^{-(m+1)},
    \end{equation*}
    where $C$ depends on $\sigma_0$ but not on $\sigma$. 
\end{lemma}
\begin{proof}
    For $j\neq 2$, Lemmas \ref{lemma:s-derivatives} and \ref{lemma:analytic-multiplication} tell us that
    \begin{equation*}
        L^{(\neq 2)}\phi =  \ac{1}\phi  +  \ac{3} \pd{\phi}{\theta} +  \ac{4} \pd{\phi}{s} +  \ac{5} \pd{^2\phi}{\theta^2} +  \ac{6} \pd{^2\phi}{\theta \partial s} +  \ac{7} \pd{^2 \phi}{s^2}
    \end{equation*}
    satisfies the analytic form and has the bound
    \begin{equation*}
        \norm{ L^{(\neq 2)}\phi}_{\sigma,H^{q}} \leq C \left(\sum_{j\neq 2}\norm{\ac{j}}_{\Sigma,\Ccal^q}\right) \norm{\phi}_{\sigma,H^{q+2D}},
    \end{equation*}
    where $C$ does not depend on $\sigma$. 
    
    For the $j=2$ term, define
    \begin{equation*}
        L^{(2)}\phi = \pd{\ac{2}}{\rho} \pd{\phi}{\rho} + \frac{1}{\rho^2} \pd{\ac{2}}{\theta}\pd{\phi}{\theta}= \pd{\ac{2}}{x} \pd{\phi}{x} + \pd{\ac{2}}{y}\pd{\phi}{y},
    \end{equation*}
    where $x = \rho \cos \theta$ and $y = \rho \sin \theta$.
    By Lemmas \ref{lemma:xy-deriv} and \ref{lemma:analytic-multiplication}, this preserves the analytic form.
    To bound the operator, first note that 
    \begin{equation*}
        \norm{\left(\pd{\phi}{\rho}\right)_m}_{H^q} = (m+1) \norm{\phi_{m+1}}_{H^{q}} \leq (m+1) \sigma^{-(m+1)} \norm{\phi}_{\sigma,H^{q+2D}}.
    \end{equation*}
    A similar bound is satisfied by the $\theta$ derivative:
    \begin{equation*}
        \norm{\left(\pd{\phi}{\theta}\right)_m}_{H^q} \leq m \sigma^{-m} \norm{\phi}_{\sigma,H^{q+2D}}.
    \end{equation*}
    So, we focus on the $\rho$ derivative term of $L^{(2)}$, where the same steps can be used to bound the $\theta$ derivative term.
    We have that
    \begin{align*}
        \norm{\left(\pd{\ac{2}}{\rho} \pd{\phi}{\rho} \right)_m}_{H^q}  &\leq \norm{\sum_{n=0}^m \left(\pd{\ac{2}}{\rho}\right)_{n+1} \left(\pd{\phi}{\rho}\right)_{m-n+1}}_{H^q}, \\
        &\leq C \norm{\ac{2}}_{\Sigma,\Ccal^q} \norm{\phi}_{\sigma,H^{q+2D}} \sum_{n=0}^m (n+1)(m-n+1) \Sigma^{-(n+1)} \sigma^{-(m-n+1)}, \\
        & \leq C (m+1) \norm{\ac{2}}_{\Sigma,\Ccal^q} \norm{\phi}_{\sigma,H^{q+2D}} \sigma^{-(m+1)} \sum_{n=0}^{m+1} (n+1) \left(\frac{\sigma_0}{\Sigma}\right)^n,\\
        & \leq C (m+1) \norm{\ac{2}}_{\Sigma,\Ccal^q} \norm{\phi}_{\sigma,H^{q+2D}} \sigma^{-(m+1)}.
    \end{align*}
    Combining the estimates on $L^{(2)}$ and $L^{(\neq 2)}$, we have our theorem.
\end{proof}

Then, the main step in proving Thm.~\ref{thm:reg-theorem} is to show the inductive step in Lemma \ref{lemma:inductive}.
For this, we depend on the following interior regularity theorem for the regularization:
\begin{theorem}[{\citet[Theorem 11.1]{taylor2011}}]
\label{thm:interior-regularity}
    If $P$ is elliptic of order $2 D$ and $u\in\mathcal{D}'(M)$, $Pu = h \in H^{q}(M)$, then $u\in H_{\mathrm{loc}}^{q+2D}(M)$, and, for each $U \subset\subset V \subset\subset M$, $\sigma<q+2D$, there is an estimate
    \begin{equation*}
        \norm{u}_{H^{q+2D}(U)} \leq C \norm{P u}_{H^{q}(V)} + C \norm{u}_{H^\sigma(V)}.
    \end{equation*}
\end{theorem} 
Then, our inductive step is:
\begin{lemma}
\label{lemma:inductive}
    Assume the hypotheses of Thm.~\ref{thm:reg-theorem} and let $\sigma \leq \sigma_0 \leq \sigma' < \Sigma$ and $m\geq 0$. 
    Suppose we have computed the finite solution
    \begin{equation*}
        \phi_{<m} = \sum_{n=0}^{m-1}\rho^n \phi_n(\theta,s), \qquad \phi_n =  \sum_{\ell = 0}^{n}  \phi_{n\ell}(s) e^{(2\ell-n)\I \theta},
    \end{equation*}
    where
    \begin{equation*}
        \norm{\phi_n}_{H^{q+2D}} < \Phi \sigma^{-n} \qquad \text{ for all } n<m.
    \end{equation*}
    Then,
    \begin{equation}
    \label{eq:phim-form}
        \phi_m = \sum_{n=0}^m \phi_{mn}(s) e^{(2n-m)\I\theta}
    \end{equation}
    and there exists two constant $C_1, C_2 > 0$ independent of $m$, $\sigma$, $\sigma_0$, $\phiic$, $f$, and $\Phi$ where $C_1$ depends continuously on $a^{(j)}$ and $\sigma'$ and $C_2$ is independent of $a^{(j)}$ such that
    \begin{align*}
        \norm{\phi_m}_{H^{q+2D}} &<  C_{1} \Phi \sigma^{-m+1} + \lVert\phiic_m\rVert_{\sigma,H^{q+2D}} + \frac{C_2}{m+1} \norm{f_m}_{H^q},\\
        &< C_{1} \Phi \sigma^{-m+1} + (\lVert\phiic\rVert_{\sigma_0,H^{q+2D}} + C_2 \norm{f}_{\sigma_0,H^q}) \sigma_0^{-m}.
    \end{align*}
\end{lemma}
\begin{proof}
Let $L^{(a)}$ be as in \eqref{eq:La}. 
Then, the near-axis iteration is given by
\begin{equation*}
    \phi_m = \phiic_m + ((1+P)^{-1} \Delta_\perp^{+}  (f+L^{(a)}\phi_{<m}))_m, 
\end{equation*}
where $\Delta_\perp^+$ is defined in \eqref{eq:Delta-plus}.
As such, the triangle inequality gives
\begin{equation*}
    \norm{\phi_m}_{H^{q+2D}} \leq \norm{\phiic_m}_{H^{q+2D}} + \norm{(1+P)^{-1} (\Delta_\perp^{+} (f+L^{(a)}\phi_{<m}))_m}_{H^{q+2D}}.
\end{equation*}
For the initial conditions, we have
\begin{equation*}
    \norm{\phiic_m}_{H^{q+2D}} \leq \norm{\phiic_m}_{\sigma_0,H^{q+2D}} \sigma_0^{-m},
\end{equation*}
so we just need to focus on the second term.

Let $g = f + L^{(a)} \phi_{<m}$. By Lemma \ref{lemma:La}, we have that 
\begin{equation*}
    \norm{g_m}_{H^q} \leq \norm{f}_{\sigma_0,H^q}\sigma_0^{-m} + C (m+1) \left(\sum_{j=1}^7 \norm{\ac{j}}_{\Sigma,\Ccal^q} \right) \Phi \sigma^{-(m+1)},
\end{equation*}
where $C$ depends on $\sigma'$ and we have used $\norm{\phi_{<m}}_{\sigma,H^q} \leq \Phi$.
Next, we establish a bound for the inverse polar Laplacian. We have that
\begin{equation*}
    \left(\Delta_\perp^{+}g\right)_m = \sum_{n=1}^{m-1} \frac{1}{m^2 - (2n-m)^2}g_{m-2,n-1} e^{(2n-m)\I \theta},
\end{equation*}
so for $m\geq 2$
\begin{align}
\nonumber
    \norm{\left(\Delta_\perp^{+}g\right)_m}^2_{H^{q}} =& \sum_{n=1}^{m-1} \frac{1}{m^2 - (2n-m)^2}\norm{g_{m-2,n-1} e^{(2n-m)\I \theta}}^2_{H^{q}}, \\
\nonumber
    \leq& \frac{1}{4(m-1)} \sum_{n=1}^{m-1} \norm{g_{m-2,n-1} e^{(2n-m)\I \theta}}^2_{H^{q}}, \\
\nonumber
    =& \frac{1}{4(m-1)}\norm{g_{m-2}}^2_{H^{q}}, \\
\label{eq:g-bound}
    \leq& \frac{C^{(1)}}{m+1} \norm{f}_{\sigma_0,H^q}\sigma_0^{-m} + C^{(2)} \left(\sum_{j=1}^7 \norm{\ac{j}}_{\Sigma,\Ccal^q} \right) 
    \Phi \sigma^{-m+1},
\end{align}
where we used constants so that this is trivially extended to $m \in \{0,1\}$ where $(\Delta_\perp^+ g)_0 = (\Delta_\perp^+ g)_1 = 0$.

Now, we would like to bound the inverse operator $(1+P)^{-1}$. 
Specifically, we need it to be the case that
\begin{equation}
\label{eq:P-regularity}
    \norm{u}_{H^{q+2D}} < C \norm{(1+P) u}_{H^{q}}.
\end{equation}
We will prove this is true by the standard argument.
For the sake of contradiction, suppose that there exists a sequence $u^{(n)} \in H^{q + 2 D}$ such that $\norm{u^{(n)}}_{L^2} = 1$ and $\norm{(1+P) u^{(n)}}_{H^{q}} \to 0$. 
By Theorem \ref{thm:interior-regularity}, this tells us that $\norm{u^{(n)}}_{H^{q+2D}}\leq 1 + \epsilon$ for some $\epsilon >0$. 
By Rellich's theorem \citep[Proposition 3.4]{taylor2011}, $H^{q + 2D}$ is compactly embedded in $H^{q+2D-1}$, so there exists a subsequence such that $u_{n} \to u$ in $H^{q+2D-1}$.
This implies that $(1+P) u^{(n)} \to (1+P) u = 0$ in $H^{q-1}$, where we are using the assumption that $q\geq 1$. 
However, because $1+P$ is positive, it does not have a kernel, so the bound \eqref{eq:P-regularity} must hold.

Equation \eqref{eq:P-regularity} tells us that $1+P$ is one-to-one from $H^{q+2D}$ to $H^q$. 
Moreover, because $1+P$ is positive and self-adjoint, it must be surjective, so it is invertible and we have the inequality
\begin{equation*}
    \norm{(1+P)^{-1}\left(\Delta_\perp^{+}g\right)_m}^2_{H^{q+2D}} \leq \frac{C_1}{m+1} \norm{f}_{\sigma_0,H^q}\sigma_0^{-m} + C_2 \left(\sum_{j=1}^7 \norm{\ac{j}}_{\Sigma,\Ccal^q} \right) 
    \Phi \sigma^{-m+1},
\end{equation*}
proving the lemma.
\end{proof}

We are now ready to prove Theorem \ref{thm:reg-theorem}.
For continuity (boundedness) of the solution with respect to $\phiic$ and $f$, we choose a $\sigma$ such that $\sigma C_1 < 1$ and $\sigma \leq \sigma_0$. 
Then, let $0 < \gamma < 1 - C_1 \sigma$. We choose $\Phi$ such that 
\begin{equation*}
    C_1 \sigma \Phi + \norm{\phiic}_{\sigma,H^{q+2D}} + C_2 \norm{f}_{\sigma_0,H^q} = (1-\gamma) \Phi.
\end{equation*}
Then, we perform induction. 
At $m=0$, $\phi_{<m} = 0$, so we have trivially satisfied the initial case. 
For the inductive step, because $\sigma_0^{-m}\leq \sigma^{-m}$, we have 
\begin{equation*}
    \norm{\phi_m}_{H^{q+2D}} \leq (1-\gamma)\Phi \sigma^{-m},
\end{equation*}
implying 
\begin{equation*}
    \norm{\phi}_{\sigma, H^{q+2D}} \leq \Phi = \frac{1-\gamma}{1-\gamma - C_1 \sigma}\left(\norm{\phiic}_{\sigma,H^{q+2D}} + C_2 \norm{f}_{\sigma_0,H^q} \right).
\end{equation*}

For continuity with respect to the coefficients $A = (\ac{1}, \dots, \ac{7})$, consider fixing $\sigma$ and $\gamma$ as before. 
Then, because $C_1$ is continuous with respect to $A$, there is a small enough perturbation so that both $\sigma$ and $\gamma$ continue to satisfy $0 < \gamma < 1-C_1 \sigma$ and $C_1 \sigma < 1$. 
So, there is a neighborhood of $U$ of $A$ such that for all $A+\delta A \in U$ and some value of $C$, we have $\norm{\phi}_{\sigma,H^{q+2D}} \leq C (\lVert \phiic \rVert_{\sigma_0,H^{q+2D}} + \norm{f}_{\sigma_0,H^{q}})$. 
Because $C_1$ does not depend on $\sigma_0$, it is also the case that there is a neighborhood of $U_0$ of $A$ such that $\norm{\phi}_{\sigma,H^{q+2D}}\leq C_0 (\lVert \phiic \rVert_{\sigma,H^{q+2D}} + \norm{f}_{\sigma_0,H^{q}})$ for all $A+\delta A \in U_0$. 

Now, consider the full PDE operator
\begin{equation*}
    L =  \Delta_\perp (1+P) - L^{(a)},
\end{equation*}
and let $L+\delta L$ be the operator associated with substituting $\ac{j}$ with $\ac{j}+\delta \ac{j}\in U \cap U_0$, i.e.
\begin{multline*}
    \delta L\phi = \delta L^{(a)}\phi = \delta \ac{1}\phi + \pd{(\delta \ac{2})}{\rho} \pd{\phi}{\rho} + \frac{1}{\rho^2} \pd{(\delta \ac{2})}{\theta}\pd{\phi}{\theta} + \\ \delta \ac{3} \pd{\phi}{\theta}+ \delta \ac{4} \pd{\phi}{s} + \delta \ac{5} \pd{^2\phi}{\theta^2} + \delta \ac{6} \pd{^2\phi}{\theta \partial s} + \delta \ac{7} \pd{^2 \phi}{s^2}.
\end{multline*}
With fixed initial conditions $\phiic$ and $f$, let the solution to the original PDE be $\phi$ (i.e.~$L \phi = f$) and the solution to the perturbed PDE be $\phi + \delta \phi$ (i.e.~$(L+\delta L)(\phi + \delta \phi) = f$). 
Subtracting the two PDE formulas gives
\begin{equation}
\label{eq:perturbed-ODe}
    (L+\delta L)\delta \phi = -\delta L \phi,
\end{equation}
where $\delta \phi$ satisfies $\delta \phiic = 0$.
Using Lemma \ref{lemma:La}, we have that 
\begin{equation*}
    \norm{(\delta L \phi)_m}_{H^q} \leq C (m+1) \left(\sum_{j=1}^7 \norm{\delta \ac{j}}_{\Sigma,\Ccal^q} \right) \norm{\phi}_{\sigma,H^{q+2D}} \sigma^{-m}.
\end{equation*}
Then, we find that 
\begin{equation*}
    C_1 \Phi \sigma^{-m+1} + \frac{C_2}{m+1}\norm{f_m}{H^q} \leq \left(C_1 \sigma + C \left(\sum_{j=1}^7 \norm{\delta \ac{j}}_{\Sigma,\Ccal^q} \right) \right) \sigma^{-m}.
\end{equation*}
We can take $\delta A$ small enough so the right term is less than $1$, allowing us to proceed inductively as before, giving continuity in the coefficients. 

\subsection{Proof of Corollary \ref{cor:convergence}}
\label{subsec:convergence}
For the coefficients of the PDE, we must only notice that
\begin{equation*}
    \frac{1}{1-\kappa \rho \cos \theta} = \sum_{n=1}^{\infty}\rho^n \left(\kappa \cos \theta\right)^n.
\end{equation*}
Because $\bm r_0 \in \Ccal^{4+q}(\Tbb)$, $\kappa \in \Ccal^{3+q}(\Tbb)$.
This immediately tells us that 
\begin{equation*}
    \norm{(\kappa \cos \theta)^n}_{\Ccal^{3+q}} \leq \left(C \norm{\kappa \cos \theta}_{\Ccal^{3+q}}\right)^n
\end{equation*}
for some constant $C$, showing this converges. (In fact, the sum converges for all $\rho < \kappa$).

We note that $P$ satisfies the Hypotheses \ref{hypotheses} by construction, so we can apply Theorem \ref{thm:reg-theorem}.
The regularity of $\bm B_K$ is the obtained from \ref{eq:BK}, where the reduction in regularity comes from the order of $P$, combined with Lemmas \ref{lemma:xy-deriv} and \ref{lemma:analytic-multiplication}.

\subsection{Proof of Propsition \ref{prop:a-posteriori}}
\label{subsec:a-posteriori}
For the statement about uniqueness, suppose $\phi' \in H^{q+2D}(\Omega_{\sigma'}^0)$ with $\phi'-\phi=\delta \phi \neq 0$ satisfies the boundary value problem.
Then, $\delta \phi$ satisfies 
\begin{equation*}
    -(P \Delta_\perp + L)\delta \phi = 0, \qquad \left. \delta \phi\right|_{\rho=\sigma'} = 0.
\end{equation*}
Then, after some algebra, we have
\begin{equation}
\label{eq:positivity}
    \ip{\pd{\delta \phi}{\rho}}{P \pd{\delta \phi}{\rho}} + \ip{\frac{1}{\rho} \pd{\delta \phi}{\theta}}{P \frac{1}{\rho} \pd{\delta \phi}{\rho}} - \ip{\delta \phi}{L \delta \phi} = 0,
\end{equation}
where the inner product is the $L^2$ inner product
\begin{equation*}
    \ip{f}{g} = \int_0^{2\pi} \int_{0}^{2\pi} \int_{0}^{\sigma'} f(\rho,\theta,s) g(\rho,\theta,s) \, \rho\df\rho\df\theta\df s.
\end{equation*}
By our assumptions on $P$ and $L$, each term in \eqref{eq:positivity} is positive.
So, it then must be the case that $\ip{\delta \phi}{L\delta \phi} = 0$. 
However, this is only possible when $\delta \phi = 0$, so the solution is unique.

For the error estimate, fix $\phi$ and subtract the two boundary value problems to find that
\begin{equation*}
    L(\phi - \tilde \phi) = - P \Delta_\perp \phi, \qquad \left.(\phi - \tilde \phi)\right|_{\rho = \sigma'} = 0.
\end{equation*}
Because $L$ is negative, there is a unique solution $\phi-\tilde \phi$ to this problem. 
Then, by standard regularity theory \citep[Theorem 6.3.5]{evans2010}, we have the desired bound.

\section{Asymptotic Expansions of Basic operations}
\label{app:operations}
In order to build the near-axis code, we need some facts from formal expansions. We let $A$, $B$, and $C$ be smooth formal power series of the generic form
\begin{equation*}
    A(\rho,\theta,s) = \sum_{n = 0}^\infty A_n(\theta, s) \rho^n,
\end{equation*}
and let $\alpha \in \Rbb$. 
Here, we explain how we numerically perform the following basic operations:
\begin{enumerate}
    \item Multiplication (\S\ref{subsec:multiplication}): $C=AB$, 
    \item Multiplicative Inversion (\S\ref{subsec:inversion}): $B=A^{-1}$,
    \item Differentiation with respect to $\rho$ (\S\ref{subsec:rho-derivatives}): $B = \dd{A}{\rho}$,
    \item Exponentiation (\S\ref{subsec:exponentiation}): $B = e^{\alpha A}$,
    \item Power (\S\ref{subsec:power}): $B = A^{\alpha}$, 
    \item Composition (\S\ref{subsec:composition}): $A(\tilde \rho, \tilde \theta, \phi)$ where $\begin{pmatrix} B & C\end{pmatrix}^T = \begin{pmatrix} \tilde \rho \cos \tilde \theta & \tilde \rho \sin \tilde \theta\end{pmatrix}^T$.
    \item Series Inversion (\S\ref{subsec:composition}): find the inverse coordinate transformation $C^{(i)}$ of the transformation $B{(i)}$, i.e.
    \begin{equation*} 
        \begin{pmatrix}
            C^{(1)}(B^{(1)}(x,y,s),B^{(2)}(x,y,s),s) \\ C^{(2)}(B^{(1)}(x,y,s),B^{(2)}(x,y,s),s)
        \end{pmatrix} = \begin{pmatrix}
            x \\ y
        \end{pmatrix}.
    \end{equation*}
\end{enumerate}
We find that these operations build upon each other, with multiplication, $\rho$-differentiation, and series composition being the main building blocks of other, more complicated algorithms.

\subsection{Multiplication}
\label{subsec:multiplication}
The most basic problem is that of (matrix) multiplication. Let $C$ be the solution to
\begin{equation*}
    C = AB.
\end{equation*}
Via simple matching of orders, we find that
\begin{equation}
\label{eq:ps_multiplication}
    C_n(\theta, s) = \sum_{m = 0}^n A_m(\theta, s) B_{n-m}(\theta, s).
\end{equation}

\subsection{Inversion}
\label{subsec:inversion}
Now, instead consider the problem of finding the (matrix) inverse
\begin{equation*}
    B = A^{-1}.
\end{equation*}
It is easier to write this in terms of the problem
\begin{equation*}
    AB = I.
\end{equation*}
So, using \eqref{eq:ps_multiplication}, we have
\begin{equation*}
    \sum_{m = 0}^n A_m B_{n-m} = 
    \begin{cases}
        I, & n = 0 ,\\
        0, & n > 0.
    \end{cases}
\end{equation*}
Assuming $A_0$ is invertible, an iterative method for finding $B_n$ is
\begin{equation}
\label{eq:ps_inverse}
    B_n = \begin{cases}
        A_0^{-1}, & n = 0, \\
        - B_0 \left( \sum_{m=1}^n A_m B_{n-m} \right), & n > 0.
    \end{cases}
\end{equation}

\subsection{\texorpdfstring{$\rho$}{ρ} Derivatives}
\label{subsec:rho-derivatives}
Let
\begin{equation*}
    B = \dd{A}{\rho}. 
\end{equation*}
Then, we have
\begin{equation*}
    B = \sum_{n = 0}^{\infty} n A_n \rho^{n-1},
\end{equation*}
or
\begin{equation}
\label{eq:ps_epsilon_derivative}
    B_n = (n+1) A_{n+1}.
\end{equation}

\subsection{Exponentiation}
\label{subsec:exponentiation}
A more complicated series operation is scalar exponentiation \citep[see][]{knuth_art_1997}. We would like to find
\begin{equation*}
    B = e^{\alpha A}.
\end{equation*}
Taking a derivative with respect to $\rho$, we find
\begin{align*}
    \dd{B}{\rho} &= \dd{e^{\alpha A}}{A} \dd{A}{\rho}, \\
    &= \alpha e^{\alpha A} \dd{A}{\rho}, \\
    &= \alpha B \dd{A}{\rho}.
\end{align*}
Now, we can write out the multiplication using \eqref{eq:ps_multiplication} and \eqref{eq:ps_epsilon_derivative}, giving
\begin{equation*}
    (n+1) B_{n+1} = \alpha \sum_{m = 0}^{n}(m+1)A_{m+1}B_{n-m},
\end{equation*}
where $B_0 = e^{\alpha A_0}$. From this, we have
\begin{equation}
\label{eq:ps_exponentiation}
    B_n = \begin{cases}
        e^{\alpha A_0}, & n = 0, \\ 
        \frac{\alpha}{n} \sum_{m=0}^{n-1}(m+1)A_{m+1}B_{n-1-m}, & n > 0.
    \end{cases}
\end{equation}

Note that through exponentiation and inversion, we can obtain any trigonometric or hyperbolic trigonometric function. For instance, we have 
\begin{equation*}
    \sin A = \Im e^{\I A},
\end{equation*}
and 
\begin{equation*}
    \sech A = 2 \left( e^A + e^{-A}\right)^{-1}.
\end{equation*}

As a more detailed example, consider the problem of simultaneously computing $\sin(A)$ and $\cos(A)$. Using $\alpha = i \omega$ in \eqref{eq:ps_exponentiation}, we find
\begin{equation*}
    \cos(\omega A)_n + \I \sin(\omega A)_n = \begin{cases}
        \cos(\omega A_0) + \I \sin(\omega A_0), & n = 0, \\ \\ 
        - \sum_{m=0}^{n-1}\frac{\omega(m+1)}{n} A_{m+1}\sin(A)_{n-1-m} & \\
        \ \ \ \ + \, \I \sum_{m=0}^{n-1}\frac{\omega(m+1)}{n}A_{m+1}\cos(A)_{n-1-m}, & n > 0.
    \end{cases}
\end{equation*}
From this, we can work with only real-valued series via the formulas
\begin{align*}
    \cos(A)_n &= \begin{cases}
        \cos(A_0), & n = 0, \\
        - \sum_{m=0}^{n-1}\frac{\omega(m+1)}{n} A_{m+1}\sin(A)_{n-1-m}, & n > 0,
    \end{cases} \\
    \sin(A)_n &= \begin{cases}
        \sin(A_0), & n = 0, \\
        \sum_{m=0}^{n-1}\frac{\omega(m+1)}{n}A_{m+1}\cos(A)_{n-1-m}, & n > 0.
    \end{cases}
\end{align*}

\subsection{Power}
\label{subsec:power}
Now, we take a power $\alpha \neq 1$ of a series:
\begin{equation*}
    B = A^{\alpha}.
\end{equation*}
Assuming $A_0 \neq 0$, we have $B_0 = A_0^\alpha$ at leading order. 
At higher orders, we take the derivative with respect to $\rho$ to find
\begin{equation*}
    \dd{B}{\rho} = \alpha A^{\alpha-1} \dd{A}{\rho}.
\end{equation*}
Multiplying both sides against $A$, we find
\begin{equation*}
    A \dd{B}{\rho} - \alpha B \dd{A}{\rho} = 0.
\end{equation*}
Term by term, we have
\begin{equation*}
    \sum_{j = 0}^{n-1} (j+1) A_{n-1-j} B_{j+1} - \alpha (j+1) B_{n-1-j} A_{j+1} = 0.
\end{equation*}
This gives 
\begin{equation*}
    B_{n} = \frac{1}{nA_0}\left(n \alpha A_n B_0 + \sum_{j=0}^{n-2} (j+1) (-A_{n-1-j} B_{j+1} + \alpha B_{n-1-j} A_{j+1}) \right).
\end{equation*}
By substituting $j \to n-j-2$, we reorder the right sum to find
\begin{equation*}
    B_{n} = \frac{1}{nA_0}\left(n \alpha A_{n}B_0 + \sum_{j=0}^{n-2} (\alpha(n-j-1) - (j+1)) B_{j+1} A_{n-j-1} \right).
\end{equation*}

\subsection{Composition}
\label{subsec:composition}
Series composition is an important operation for changing coordinates.
Consider that $A = A(x,y,s)$, where $x=R \cos \Theta$ and $y = R \sin \Theta$.
We would like to compose this with the functions $B(\rho,\theta)$ and $C(\rho,\theta)$ as $A(B,C,s)$, where $B$ and $C$ replace the Cartesian cross-section coordinates $x$ and $y$. 
We note that this is preferable to composing with the polar coordinates $R$ and $\Theta$, as there are valid non-analytic transformations in these coordinates that keep $A$ analytic. 
For this transformation, we assume that $B_0 = C_0 = 0$, i.e.~there is no constant offset. 
We present the details for numerical Fourier series composition \eqref{eq:real-fourier-form}, but equivalent expressions could be used for other forms.

The main observation we make for series composition is that if we can compute the basis functions
\begin{equation}
\label{eq:Phi-basis}
    \Phi_{mn}(\rho,\theta,s) = R(\rho,\theta,s)^m \Fcal_{2n-m}(\Theta(\rho,\theta,s)),
\end{equation}
then the composition is simply
\begin{equation*}
    A(B,C,s) = \sum_{m=0}^{\infty} \sum_{n=0}^m A_{mn}(s) \Phi_{mn}(\rho,\theta,s),
\end{equation*}
where the multiplication in $s$ can be performed in spatial coordinates.
So, the majority of the work is to find $\Phi_{mn}$, with the added perk that once the basis is found, further series compositions are faster.

To find the basis, we first notice that
\begin{equation*}
    \Phi_{00} = 1, \qquad \Phi_{10} = R\cos\Theta = C, \qquad \Phi_{11} = R\sin\Theta = B.
\end{equation*} 
Then, further functions can be found by using angle-sum identities in the cosine case ($m/2 \leq n \leq m+1$)
\begin{align*}
    \Phi_{m+1,n} &= R^{m+1}\cos((2n-m-1)\Theta),\\
        &= R^{m+1} \cos((2(n-1)-m)\Theta)\cos(\Theta) - R^{m+1} \sin((2(n-1)-m)\Theta) \sin(\Theta),\\
        &= \Phi_{m,n-1} \Phi_{11} - \Phi_{m,m-n+1} \Phi_{10},
\end{align*}
and the sine case ($0\leq n<m/2$)
\begin{align*}
    \Phi_{m+1,n} &= R^{m+1}\sin((m + 1 - 2n)\Theta),\\
    &= R^{m+1}\sin((m-2n)\Theta)\cos(\Theta) + \cos((m-2n)\Theta)\sin(\Theta), \\
    &= \Phi_{mn}\Phi_{11} + \Phi_{m,m-n}\Phi_{10}.
\end{align*}

\subsection{Series Inversion}
\label{subsec:series-inversion}
Consider that we know the (Cartesian) flux coordinates $(B^{(1)}(x,y,s),B^{(2)}(x,y,s)) = (X,Y) = (R \cos \Theta, R \sin \Theta)$ and we want to know how to represent a function $A(\rho,\theta,s)$ in terms of $R$ and $\Theta$.
For this, we would use the series composition step presented in the previous section, but we need to obtain $\rho$ and $\theta$ in terms of $R$ and $\Theta$, i.e. we need $(C^{(1)}(X,Y,s), C^{(2)}(X,Y,s)) = (x,y) = (\rho \cos\theta, \rho \sin\theta)$, where 
\begin{equation}
\label{eq:series-inversion}
    \begin{pmatrix}
        C^{(1)}(B^{(1)},B^{(2)},s) \\ C^{(2)}(B^{(1)},B^{(2)},s)
    \end{pmatrix} = \begin{pmatrix}
        x \\ y
    \end{pmatrix}.
\end{equation}
This is a step that is necessary if we have a series represented in direct coordinates, and we want to represent it in indirect (flux) coordinates.

We will solve this equation iteratively, assuming that we are using the real Fourier form \eqref{eq:real-fourier-form}. At leading order (assuming $C^{(i)}_0 = B^{(i)}_0 = 0$), we have
\begin{equation*}
    \begin{pmatrix}
        C^{(1)}_{11} & C^{(1)}_{10} \\ C^{(2)}_{11} & C^{(2)}_{10}
    \end{pmatrix}
    \begin{pmatrix}
        B^{(1)}_{11} & B^{(1)}_{10} \\ B^{(2)}_{11} & B^{(2)}_{10}
    \end{pmatrix} \begin{pmatrix}
        x \\ y
    \end{pmatrix} = \begin{pmatrix}
        x \\ y
    \end{pmatrix},
\end{equation*}
or
\begin{equation*}
    \begin{pmatrix}
        C^{(1)}_{11} & C^{(1)}_{10} \\ C^{(2)}_{11} & C^{(2)}_{10}
    \end{pmatrix} = \begin{pmatrix}
        B^{(1)}_{11} & B^{(1)}_{10} \\ B^{(2)}_{11} & B^{(2)}_{10}
    \end{pmatrix}^{-1}.
\end{equation*}

For the next orders, we note that we can compute $C^{(i)}_{<m+1}(B^{(1)}, B^{(2)}, s)$ using the composition formula, where we have chosen $C_{m}$ so this is identity up to order $m+1$.
Substituting this into \eqref{eq:series-inversion}, we have
\begin{equation*}
    \begin{pmatrix}
        C^{(1)}_{m+1}(\rho B^{(1)}_1, \rho B^{(2)}_1, s) \\ C^{(2)}_{m+1}(\rho B^{(1)}_1, \rho B^{(2)}_1, s)
    \end{pmatrix} = - \begin{pmatrix}
        (C^{(1)}_{<m+1}(B^{(1)}, B^{(2)}, s))_{m+1} \\ (C^{(2)}_{<m+1}(B^{(1)}, B^{(2)}, s))_{m+1}
    \end{pmatrix}.
\end{equation*}
If we build a basis $\Phi_{mn}$ from $B^{(i)}_1$ (see \eqref{eq:Phi-basis} and the following procedure) where
\begin{equation*}
    \Phi_{mn} = \sum_{k=0}^{\infty} \sum_{\ell=0}^k \Phi_{mnkl} \rho^{k} \Fcal_{k-2\ell}(\theta),
\end{equation*}
the update can be expanded in index-notation as
\begin{equation*}
    \sum_{\ell=0}^{m+1} C^{(i)}_{m+1,\ell}\Phi_{m+1,\ell,m+1,n} = -(C^{(1)}_{<m+1}(B^{(1)}, B^{(2)}, s))_{m+1,n}.
\end{equation*}
Then, by inverting the matrix $\Psi_{\ell n} = \Phi_{m+1,\ell,m+1,n}$ onto the right-hand-side, we have the update for $C^{(i)}_{m+1}$.

\bibliographystyle{jpp}

\bibliography{references}

\end{document}